\documentclass[a4paper,onecolumn,accepted=2026-06-15]{quantumarticle}
\pdfoutput=1

\usepackage[utf8]{inputenc}
\usepackage[english]{babel}
\usepackage[T1]{fontenc}
\usepackage{amsmath,amssymb,amsthm}
\usepackage{hyperref}

\usepackage[numbers]{natbib}

\usepackage{thmtools, thm-restate}
\usepackage{xspace}
\usepackage{multirow}
\usepackage{makecell}
\usepackage{bm}
\usepackage{url}
\usepackage{tikz}
\usetikzlibrary{shapes}
\usetikzlibrary{cd}
\usepackage[braket]{qcircuit}
\usepackage{calc}
\usepackage{caption}
\usepackage{subcaption}
\usepackage{tikz}

\usepackage[titletoc,title]{appendix}

\DeclareMathOperator{\cutrank}{cutrank}
\DeclareMathOperator{\rw}{rw}

\usepackage{algpseudocode}

\usepackage[linesnumbered,ruled,vlined]{algorithm2e}
\SetKwProg{Procedure}{Procedure}{}{}
\SetKwInOut{Input}{Input}
\SetKwInOut{Output}{Output}

\newtheorem{theorem}{Theorem}

\newtheorem{lemma}[theorem]{Lemma}
\newtheorem{corollary}[theorem]{Corollary}
\theoremstyle{definition}
\newtheorem{definition}[theorem]{Definition}

\newtheorem{fact}[theorem]{Fact}

\theoremstyle{remark}
\newtheorem{remark}[theorem]{Remark}

\newcommand{\CZ}{\ensuremath{\mathrm{CZ}}\xspace}
\newcommand{\NCZ}{\Gamma_\mathrm{CZ}}
\newcommand{\NCO}{\Gamma_\mathrm{EC}}

\title{Complexity of graph-state preparation by Clifford circuits}
\author{Soh Kumabe}
\affiliation{CyberAgent, Japan}
\email{kumabe\_soh@cyberagent.co.jp}
\orcid{0000-0002-1021-8922}
\author{Ryuhei Mori}
\affiliation{Graduate School of Mathematics, Nagoya University, Japan}
\email{mori@math.nagoya-u.ac.jp}
\orcid{0000-0001-5474-5145}
\author{Yusei Yoshimura}
\affiliation{School of Computing, The Institute of Science Tokyo, Japan}

\begin{document}
\maketitle

\begin{abstract}
In this work, we study the complexity of graph-state preparation in a general model of quantum algorithms that allows measurements in the computational basis, single-qubit Clifford operations, and two-qubit Clifford operations.
We define the \emph{CZ-complexity} of a graph state $\ket{G}$ as the minimum number of two-qubit Clifford operations required to generate $\ket{G}$ from $\ket{0}^{\otimes (n+s)}$ for some $s\ge 0$. Equivalently, every optimal algorithm can be taken to use only controlled-Z (\CZ) gates as its two-qubit Clifford operations.
We then give a combinatorial characterization of graph-state transformations.
Specifically, $\ket{G}$ can be generated from another graph state $\ket{H}$ by an algorithm of CZ-complexity at most $t$ if and only if $G$ can be obtained from $H$ by vertex deletions, local complementations and at most $t$ \emph{elementary edge-complementations}.
Here, an elementary edge-complementation toggles either a single edge, all edges between one vertex and the neighborhood of another, or all edges between the neighborhoods of two non-adjacent vertices.
Using this characterization, we relate CZ-complexity to rank-width.
For any graph $G$ with $n$ vertices and rank-width $r$, the CZ-complexity is $O(rn)$, and if $G$ is connected then it is at least $n+r-2$.
We also show that these bounds are close to optimal.
Finally, for interval graphs and circle graphs, whose rank-width is unbounded, we present preparation algorithms with CZ-complexity $O(n)$ and $O(n\log n)$, respectively.
\end{abstract}


\section{Introduction}
The set of stabilizer states is a large class of quantum states used in quantum information processing.
The set of graph states is a subclass of stabilizer states represented by an undirected graph.
Any stabilizer state can be transformed into a graph state using local Clifford operations~\cite{PhysRevA.69.022316}.
Hence, it is important to develop efficient quantum algorithms for preparing graph states since
any stabilizer state can be transformed from some graph state by local Clifford operations.
Graph states have been studied from several perspectives, including their multipartite entanglement and entanglement purification~\cite{PhysRevLett.91.107903,PhysRevA.69.062311}, their role as resource states for measurement-based quantum computation~\cite{PhysRevA.68.022312,PhysRevA.75.012337,PhysRevA.76.022304,harrison2025fermion,hahn2026structure}, graph-theoretic descriptions of their transformation under local Clifford operations and related single-qubit operations~\cite{PhysRevA.69.022316,PhysRevA.77.042307,dahlberg2018,Adcock2020mappinggraphstate,claudet_et_al:LIPIcs.STACS.2025.27}, and optimal preparation problems~\cite{PhysRevA.83.042314,Sharma2026minimisingnumberof}.
In this work, we focus on the last aspect and study graph-state preparation complexity from the viewpoint of the number of two-qubit Clifford operations.


In general, any stabilizer state can be generated from $\ket{0}^{\otimes n}$ by an $n$-qubit Clifford operation.
Any $n$-qubit Clifford operation is generated by single-qubit and two-qubit Clifford operations.
For example, the commonly used generators of the Clifford group are the $S$, $H$ and \CZ gates.
In this work, we study the complexity of graph-state preparation.
Although there are several implementations of quantum computing, e.g., superconductors, trapped-ion, etc.,
we are interested in complexity measures that capture universal difficulties in graph-state preparation.
The implementation of two-qubit operations is generally challenging.
Therefore, we define the complexity of graph states as follows.

\begin{definition}[The CZ-complexity of graph states]\label{def:comp}
Let $\mathcal{A}$ be a quantum algorithm
that consists of the following operations.
\begin{enumerate}
\item Measurements of qubits in the computational basis.
\item Single-qubit Clifford operations.
\item Two-qubit Clifford operations.
\end{enumerate}
Here, we assume that measured qubits are discarded in quantum algorithms.
After the measurements of some qubits, the algorithm may choose subsequent operations depending on the measurement outcomes.
The \emph{\CZ-complexity of $\mathcal{A}$} is defined as the maximum number (with respect to the measurement outcome) of two-qubit Clifford operations in $\mathcal{A}$.
The \emph{\CZ-complexity $\NCZ(\ket{G})$ of graph state $\ket{G}$} is defined as the minimum CZ-complexity of $\mathcal{A}$
that generates the $n$-qubit graph state $\ket{G}$ from $\ket{0}^{\otimes (n+s)}$ with probability 1 for some $s\ge 0$.
\end{definition}

In a physical realization of a quantum computer, the use of additional working qubits may be costly.
In such a setting, the number of working qubits should also be included in the cost of a preparation algorithm.
Furthermore, real quantum computers have hardware-specific connectivity constraints, so that \CZ gates can be applied only to restricted pairs of qubits.
In this work, however, we ignore such hardware-dependent issues, and count only the number of \CZ operations.
This viewpoint allows us to study the cost of graph-state preparation in a clean, architecture-independent manner.
Motivated by applications of graph states in several settings, including measurement-based quantum computation, our goal is to understand which graph states can be prepared using only a small number of \CZ operations and how this preparation cost is governed by graph structure.
Our results give upper bounds, lower bounds, and structural characterizations for \CZ-complexity of preparing such resource states.

There is a simple quantum algorithm generating an $n$-qubit graph state $\ket{G}$.
First, $\ket{+}^{\otimes n}$ is prepared, which corresponds to the graph state for the empty graph of $n$ vertices.
Here, each qubit corresponds to each vertex of $G$.
Then, for each edge of $G$, the controlled-Z (CZ) operation is applied for the corresponding pair of qubits.
This quantum algorithm generates $\ket{G}$ with $m$ two-qubit Clifford operations where $m$ is the number of edges in $G$.

Van den Nest, Dehaene, and De Moor proved that two graph states $\ket{G}$ and $\ket{H}$ can be transformed
to each other by local Clifford operations if and only if $G$ and $H$ can be transformed to each other
by some combinatorial graph transformations, called local complementations~\cite{PhysRevA.69.022316}.
The local complementation at a vertex $v$ is an operation that complements edges for all pairs in neighborhoods of $v$.
The local complementation can reduce the number of edges in graphs in some cases.
For example, the complete graph $K_n$ can be transformed to the star graph $S_{n-1}$ by a local complementation.
This implies that $\ket{K_n}$ is generated from $\ket{S_{n-1}}$ by local Clifford operations.
Since $S_{n-1}$ has $n-1$ edges, $\ket{S_{n-1}}$ can be generated with $n-1$ two-qubit Clifford operations.
Hence, the CZ-complexity of $\ket{K_n}$ is at most $n-1$
(Indeed, this is a lower bound for any connected graph as well).
This example shows that local Clifford operations are useful for reducing the number of two-qubit Clifford operations.
As another important example, Figure~\ref{fig:c4} shows that the cycle graph $C_4$ of size four can be transformed to the path graph $P_4$ of size four by the local complementations.
Note that we do not allow the use of local non-Clifford operations in Definition~\ref{def:comp}
although they are generally stronger than local Clifford operations~\cite{lulc2010,claudet_et_al:LIPIcs.STACS.2025.27}.

\begin{figure}[t]
\begin{tikzpicture}[scale=0.25,rotate=45]
\node[circle,minimum size=15,draw,very thick] (a) at (0:4) {};
\node[circle,minimum size=15,draw,very thick] (b) at (90: 4) {};
\node[circle,minimum size=15,draw,very thick,fill=red] (c) at (180: 4) {};
\node[circle,minimum size=15,draw,very thick] (d) at (270: 4) {};
\draw[-,very thick] (a) -- (b);
\draw[-,very thick] (b) -- (c);
\draw[-,very thick] (c) -- (d);
\draw[-,very thick] (d) -- (a);
\end{tikzpicture}
\hfill
\begin{tikzpicture}[scale=0.25,rotate=45]
\node[circle,minimum size=15,draw,very thick] (a) at (0:4) {};
\node[circle,minimum size=15,draw,very thick,fill=red] (b) at (90: 4) {};
\node[circle,minimum size=15,draw,very thick] (c) at (180: 4) {};
\node[circle,minimum size=15,draw,very thick] (d) at (270: 4) {};
\draw[-,very thick] (a) -- (b);
\draw[-,very thick] (b) -- (c);
\draw[-,very thick] (c) -- (d);
\draw[-,very thick] (d) -- (a);
\draw[-,very thick] (b) -- (d);
\end{tikzpicture}
\hfill
\begin{tikzpicture}[scale=0.25,rotate=45]
\node[circle,minimum size=15,draw,very thick] (a) at (0:4) {};
\node[circle,minimum size=15,draw,very thick] (b) at (90: 4) {};
\node[circle,minimum size=15,draw,very thick,fill=red] (c) at (180: 4) {};
\node[circle,minimum size=15,draw,very thick] (d) at (270: 4) {};
\draw[-,very thick] (a) -- (b);
\draw[-,very thick] (b) -- (c);
\draw[-,very thick] (b) -- (d);
\draw[-,very thick] (a) -- (c);
\end{tikzpicture}
\hfill
\begin{tikzpicture}[scale=0.25,rotate=45]
\node[circle,minimum size=15,draw,very thick] (a) at (0:4) {};
\node[circle,minimum size=15,draw,very thick] (b) at (90: 4) {};
\node[circle,minimum size=15,draw,very thick] (c) at (180: 4) {};
\node[circle,minimum size=15,draw,very thick] (d) at (270: 4) {};
\draw[-,very thick] (b) -- (c);
\draw[-,very thick] (b) -- (d);
\draw[-,very thick] (a) -- (c);
\end{tikzpicture}
\caption{$C_4$ and $P_4$ are local-complement equivalent. When the local complementations are applied for the filled vertices, $C_4$ is transformed to $P_4$.}
\label{fig:c4}
\end{figure}

\begin{figure}[t]
\hfill
\begin{tikzpicture}
\node[circle,draw,very thick] (a) at (90:0.4) {};
\node[circle,draw,very thick] (b) at (-150:0.4) {};
\node[circle,draw,very thick] (c) at (-30:0.4) {};
\node[circle,draw,very thick] (e) at (90:1) {};
\node[circle,draw,very thick] (f) at (-150:1) {};
\node[circle,draw,very thick] (g) at (-30:1) {};
\draw[very thick] (a) -- (b);
\draw[very thick] (b) -- (c);
\draw[very thick] (c) -- (a);
\draw[very thick] (a) -- (e);
\draw[very thick] (b) -- (f);
\draw[very thick] (c) -- (g);
\end{tikzpicture}
\hfill
\begin{tikzpicture}
\node[circle,draw,very thick] (a) at (0:0.7) {};
\node[circle,draw,very thick] (b) at (90:0.7) {};
\node[circle,draw,very thick] (c) at (180:0.7) {};
\node[circle,draw,very thick] (d) at (-90:0.7) {};
\node[circle,draw,very thick] (e) at (0:1.6) {};
\node[circle,draw,very thick] (f) at (180:1.6) {};
\draw[very thick] (a) -- (b);
\draw[very thick] (b) -- (c);
\draw[very thick] (c) -- (d);
\draw[very thick] (d) -- (a);
\draw[very thick] (a) -- (e);
\draw[very thick] (c) -- (f);
\end{tikzpicture}
\hfill
~
\caption{Graphs that have the minimum number of edges among all graphs obtained by local complementations, and include cycles of size three or four.}
\label{fig:min}
\end{figure}

Based on the above observations, we can consider a general strategy for designing low CZ-complexity quantum algorithms.
To generate $\ket{G}$, find a graph $H$ that has the minimum number of edges among all graphs
that can be transformed to $G$ by local complementations.
Then, $\ket{H}$ can be generated with $m$ two-qubit operations where $m$ is the number of edges of $H$.
Finally, $\ket{G}$ can be generated from $\ket{H}$ by a local Clifford operation.
The CZ-complexities of the algorithms based on this strategy were calculated for all graph states up to 12 qubits in~\cite{PhysRevA.83.042314}.
Recently, algorithms for minimizing the number of edges by local complementations were provided in~\cite{Sharma2026minimisingnumberof}.
The algorithm based on simulated annealing efficiently finds a graph with approximately minimum number of edges up to 100 qubits.
Another algorithm exactly finds a graph with minimum number of edges up to 16 qubits.
In quantum algorithms of this type, the CZ operations are applied for $\ket{+}^{\otimes n}$ first, and then
a local Clifford operation is applied.

In fact, this simple strategy is not necessarily optimal.
The two graphs in Figure~\ref{fig:min} each have six edges, which is the minimum number possible under local-complementations~\cite{Adcock2020mappinggraphstate}.
However, these graphs can be generated with five \CZ operations since cycles of size three and four can be generated with two and three \CZ operations, respectively.\footnote{Interestingly, there exist graphs that have the minimum number of edges among all graphs obtained by local complementations, and do not contain cycles of size three or four, but can be generated more efficiently, e.g., No.131 in~\cite{Adcock2020mappinggraphstate}.}
The class of quantum algorithms in Definition~\ref{def:comp} is more general than quantum algorithms from the above strategy since any two-qubit Clifford operations rather than the CZ operation are allowed, and single-qubit Clifford operations and two-qubit Clifford operations can be applied in arbitrary order.
We first observe that, without loss of generality, we can assume that all two-qubit Clifford operations in graph-state preparation algorithms are \CZ operations, thus justifying the terminology of CZ-complexity introduced in Definition~\ref{def:comp}.
\begin{restatable}{proposition}{propcz}
\label{prop:cz0}
For any graph $G$, there exists a quantum algorithm $\mathcal{A}$ generating $\ket{G}$
of the type in Definition~\ref{def:comp} with CZ-complexity $\NCZ(\ket{G})$
where all two-qubit Clifford operations in $\mathcal{A}$ are the \CZ operations.
\end{restatable}

Proposition~\ref{prop:cz0} is proved in Appendix~\ref{apx:cz}.
For a graph state $\ket{G}$, a \CZ operation corresponds to toggling a single edge of $G$.
However, in the framework of quantum algorithms in Definition~\ref{def:comp}, \CZ operations can be performed on a stabilizer state which may not be a graph state.
We prove a combinatorial characterization of the graph state transformation with CZ-complexity at most $t$.
\begin{theorem}\label{thm:c1}
For any graphs $G$ and $H$, $\ket{G}$ can be generated from $\ket{H}$ with CZ-complexity at most $t$
if and only if $G$ can be obtained from $H$ by vertex deletions, local complementations and at most $t$ elementary edge-complementations, which are defined in Definition~\ref{def:cost1}.
\end{theorem}
This combinatorial characterization can be regarded as a generalization of the combinatorial characterization for $t=0$~\cite{PhysRevA.69.022316,dahlberg2018}.
Theorem~\ref{thm:c1} is derived from the graphical description of \CZ operations on stabilizer states~\cite{PhysRevA.77.042307}.

It is known that the asymptotic CZ-complexity is $\Theta(n^2/\log n)$~\cite{markov2008optimal}.
For better upper bounds, we consider graphs with bounded rank-width.
The rank-width is a complexity measure of graphs introduced by Oum and Seymour~\cite{OUM2006514,OUM201715}.
We derived the upper and lower bounds of the CZ-complexity using the rank-width.

\begin{restatable}{theorem}{thmupper}\label{thm:upper}
Let $G$ be a graph with $n$ vertices and rank-width at most $r\ge 1$.
If $r$ is odd, for any $n\ge \frac{13r^2-6r-3}{4r}$,
\begin{align*}
\NCZ(\ket{G})\le \frac{5r^2-1}{4r}n - \frac{221r^4-180r^3+10r^2+36r+9}{96r^2}.
\end{align*}
If $r$ is even, for any $n\ge \frac{13r-6}4$,
\begin{align*}
\NCZ(\ket{G})\le \frac{5r}4n - \frac{221r^2-180r+100}{96}.
\end{align*}
\end{restatable}
For the case $r=1$, this upper bound is optimal for connected graphs.
\begin{corollary}\label{cor:r1}
For any graph $G$ with $n$ vertices and rank-width 1,
\begin{align*}
\NCZ(\ket{G})\le n-1.
\end{align*}
\end{corollary}
Since $n-1$ is a trivial lower bound on the \CZ-complexity for any connected graph $G$, Corollary~\ref{cor:r1} implies that $\NCZ(\ket{G})=n-1$ for any connected graph $G$ with rank-width 1.
The upper bound $O(rn)$ is nearly optimal, as demonstrated by the following lemma derived from information-theoretic counting arguments.
\begin{restatable}{lemma}{leminf}
\label{lem:inf}
There exists a positive constant $c$ such that
for any positive integer $n$ and $r\le \lceil n/3\rceil$,
there exists a graph $G$ with $n$ vertices and rank-width at most $r$ such that $\NCZ(\ket{G})\ge crn/\log n$.
\end{restatable}

The lower bound derived from the rank-width is formulated as follows.
\begin{restatable}{theorem}{thmlower}\label{thm:lower}
For any connected graph $G$ with $n\ge 2$ vertices and rank-width at least $r$, $\NCZ(\ket{G})\ge n+r-2$.
\end{restatable}

The cycle graph $C_n$ with $n$ vertices and $n$ edges has rank-width two when $n\ge 5$.
Theorem~\ref{thm:lower} implies that $\NCZ(\ket{C_n})=n$ when $n\ge 5$.
This is in contrast to the fact $\NCZ(\ket{C_n})=n-1$ for $n=3,4$.
From the analysis of the rank-width of random sparse graphs~\cite{https://doi.org/10.1002/jgt.20620}, the following lemma is obtained.

\begin{restatable}{lemma}{lemtlower}
\label{lem:tlower}
There exists a constant $c\ge 1$ such that
for any $r'$ there exist $n'$ and $r\ge r'$
such that for any $n\ge n'$,
there exists a connected graph $G$ with $n$ vertices and rank-width $r$ satisfying $\NCZ(\ket{G})\le n - 1 + c (r-1)$.
\end{restatable}

This lemma implies that the lower bound $\NCZ(\ket{G})-(n-1)\ge r-1$ for a connected graph $G$ obtained from Theorem~\ref{thm:lower} is optimal up to a constant factor.

It is natural to ask whether the rank-width is an appropriate choice of the complexity measure of graphs for bounding the CZ-complexity of graph states.
A similar upper bound to Theorem~\ref{thm:upper} may be obtained from other complexity measures satisfying some properties required for the proof of Theorem~\ref{thm:upper}.
On the other hand, for the proof of Theorem~\ref{thm:lower}, the invariance under the local complementations is required for the complexity measure.
Deriving lower bounds for the CZ-complexity from other graph complexity measures is challenging due to this highly restrictive requirement, despite the upper and lower bounds in Theorems~\ref{thm:upper} and \ref{thm:lower} not being tight.

Finally, we present efficient quantum algorithms for special classes of graphs with unbounded rank-widths.

\begin{theorem}\label{thm:coin}
For any interval graph $G$ with $n$ vertices, $\NCZ(\ket{G})\le 2n-2$.
For any circle graph $G$ with $n$ vertices, $\NCZ(\ket{G})\le 1.262\cdot (n-1)\log_2 (n+1)$.
\end{theorem}

Recently, Davies and Jena proved that for any circle graph $G$ with $n$ vertices and the chromatic number $k$, $\NCZ(\ket{G}) \le 2n\lceil\log_2 k\rceil$~\cite{davies2025preparing, jena2024graph}.
The class of circle graphs plays a key role within vertex-minor theory, analogous to the role of planar graphs within graph minor theory~\cite{KIM202454}.
Davies and Jena concurrently and independently introduced the notion of \CZ-distance which is almost the same as the \CZ-complexity~\cite{davies2025preparing, jena2024graph}.
Compared to the \CZ-complexity, the \CZ-distance has two restrictions: (1) the use of working qubits that are measured and discarded is not allowed, and (2) \CZ-operations may be applied only to graph states, and not to general stabilizer states.
The second restriction incurs an overhead of at most a factor of two since only one of the three types of operations appearing in elementary edge-complementations in Definition~\ref{def:cost1} is allowed.
We do not know how large the overhead is when the use of working qubits is forbidden.
Note that all graph-state preparation algorithms in this paper do not require working qubits.
Indeed, we do not have any example in which the use of working qubits helps to reduce the number of \CZ operations.
However, we prefer to define the \CZ-complexity so as to allow the use of working qubits since the lower bounds in this paper hold even in this setting, and since the \CZ-complexity has the monotonicity property, i.e., $\NCZ(\ket{H})\le\NCZ(\ket{G})$ for any graph $G$ and any vertex-minor $H$ of $G$.


\subsection{Related work}
The graph state was introduced in~\cite{PhysRevLett.91.107903}
 and has been extensively studied~\cite{PhysRevA.68.022312,PhysRevA.69.062311,PhysRevA.69.022316,PhysRevA.77.042307,claudet_et_al:LIPIcs.STACS.2025.27}.

Before the introduction of graph states, Bouchet introduced and studied \emph{the isotropic system}, which is technically very similar to the stabilizer states~\cite{BOUCHET1987231,BOUCHET198858,bouchet1989connectivity,Bouchet1990,bouchet1991efficient,BOUCHET199375,BOUCHET1994107}.
Important notions in graph states, e.g., the local complementation, the pivoting, the vertex-minor, etc., have been studied in the literature.
As an important result, Bouchet presented an $O(n^4)$-time algorithm to decide whether two graphs can be transformed to each other by the local complementations~\cite{bouchet1991efficient}.
The notion of the local complementation was introduced by Kotzig in the context of Eulerian cycle of 4-regular graphs~\cite{MR0248043}.
Fon-Der-Flaass also showed some properties of the local complementations~\cite{Fon-Der-Flaass1996}.

Oum and Seymour introduced the rank-width~\cite{OUM2006514,OUM201715}.
The rank-width is based on the cut-rank function, which was introduced and called \emph{the connectivity function} by Bouchet~\cite{doi:10.1137/0608028,BOUCHET199375}.
Although rank-width was originally introduced for approximating clique-width, it possesses its own interesting features and has since been studied extensively~\cite{OUM200579,OUM2007385,OUM201715,https://doi.org/10.1002/jgt.20620,OUM2009745,COURCELLE200791,doi:10.1137/22M153937X,10.1145/3618260.3649732,korhonen2026branch}.


H{\o}yer, Mhalla, and Perdrix introduced a graph-theoretical measure, \emph{local minimum degree}, which is the minimum of the minimum degree up to the local complementations~\cite{hoyer2006resources}.
The local minimum degree has been studied in~\cite{10.1007/978-3-642-34611-8_16, 10.1007/978-3-662-48971-0_23, 10.1007/978-3-031-75409-8_10}.
Indeed, the local minimum degree is relevant to our algorithms in this paper.

H{\o}yer, Mhalla, and Perdrix also proved that any graph state can be prepared with a constant-depth Clifford circuit~\cite{hoyer2006resources}.

\subsection{Organization}
The remainder of the paper is organized as follows. Section~\ref{sec:stab} introduces notation and terminology for stabilizer states and the Clifford group.
Section~\ref{sec:graph0} reviews graph states and the combinatorial characterizations of relevant local Clifford operations and measurements from previous work.
Section~\ref{sec:graph1} presents the combinatorial characterization of two-qubit Clifford operations and proves Theorem~\ref{thm:c1}.
Section~\ref{sec:rank} discusses rank-width and its known properties, and also provides the proofs for Lemmas~\ref{lem:inf} and \ref{lem:tlower}. Finally, Sections~\ref{sec:lower}, \ref{sec:upper}, and \ref{sec:special} are dedicated to proving Theorems \ref{thm:lower}, \ref{thm:upper}, and \ref{thm:coin}, respectively.

\section{Stabilizer states and Clifford group}\label{sec:stab}
\subsection{Stabilizer states}
\begin{definition}[Pauli group]
The Pauli matrices are defined as
\begin{align*}
I_2 &:= \begin{bmatrix}
1&0\\
0&1
\end{bmatrix},&
X &:= \begin{bmatrix}
0&1\\
1&0
\end{bmatrix},&
Y &:= \begin{bmatrix}
0&-i\\
i&0
\end{bmatrix},&
Z &:= \begin{bmatrix}
1&0\\
0&-1
\end{bmatrix}.
\end{align*}
The Pauli group $\mathcal{G}_n$ is defined as
\begin{align*}
\mathcal{G}_n &:= \bigl\{\alpha h_1\otimes h_2\otimes\dotsm\otimes h_n\mid \alpha \in\{\pm1,\pm i\},\quad 
h_j\in \{I_2,X,Y,Z\}\ \forall j\in\{1,2,\dotsc,n\}\bigr\}.
\end{align*}
\end{definition}

\begin{definition}[Stabilizer state]
Let $S$ be a subset of $\mathcal{G}_n$ of size $n$ satisfying
\begin{enumerate}
\item $\forall s\in S,\quad s^2 = I_2^{\otimes n}$.
\item $\forall s, t\in S,\quad st=ts$.
\item For $T\subseteq S,\quad \prod_{s\in T}s = \pm I_2^{\otimes n}$ only when $T=\varnothing$.
\end{enumerate}
Then, there exists a one-dimensional subspace
$\bigl\{\ket{\psi}\in\mathbb{C}^{2^n}\mid s\ket{\psi}=\ket{\psi}\ \forall s\in S\bigr\}$.
Here, the quantum state in the subspace is called a \emph{stabilizer state
that is stabilized by the stabilizer $\langle S\rangle$}. 
Furthermore, $S$ is called \emph{the stabilizer generator}.
\end{definition}
From the condition $s^2=I_2^{\otimes n}$ for stabilizer generators, a stabilizer $\langle S\rangle$ must be a subset of
\begin{align*}
\mathcal{G}_n^\pm &:= \bigl\{\alpha h_1\otimes h_2\otimes\dotsm\otimes h_n\mid \alpha \in\{\pm1\},\quad 
h_j\in \{I_2,X,Y,Z\}\ \forall j\in\{1,2,\dotsc,n\}\bigr\}.
\end{align*}
If $S=\{s_1,\dotsc, s_n\}\subseteq \mathcal{G}_n^\pm$ is a stabilizer generator,
then for any $v\in\{0,1\}^n$, $\{(-1)^{v_1}s_1,\dotsc, (-1)^{v_n}s_n\}\subseteq \mathcal{G}_n^\pm$ is also a stabilizer generator for some stabilizer state.
Stabilizers without the sign information are represented by $n\times 2n$ binary matrices.

\begin{definition}[The binary representation of Pauli operators~\cite{gottesman1997stabilizer,PhysRevA.69.022316}]
There is an isomorphism between
$\mathcal{G}_n/\langle i I_2^{\otimes n}\rangle$ and $(\mathbb{Z}/2\mathbb{Z})^{2n}$.
For $n=1$, we define an isomorphism as
\begin{align*}
\sigma_{00}:=I_2 &\longleftrightarrow \begin{bmatrix}0&0\end{bmatrix},&
\sigma_{10}:=X &\longleftrightarrow \begin{bmatrix}1&0\end{bmatrix},&
\sigma_{11}:=Y &\longleftrightarrow \begin{bmatrix}1&1\end{bmatrix},&
\sigma_{01}:=Z &\longleftrightarrow \begin{bmatrix}0&1\end{bmatrix}.
\end{align*}
For general $n$, we define an isomorphism as
\begin{align*}
\sigma_{x_1z_1}\otimes \sigma_{x_2z_2}\otimes\dotsm \otimes\sigma_{x_nz_n} &\longleftrightarrow
\begin{bmatrix}
x_1&
\cdots&
x_n&
z_1&
\cdots&
z_n
\end{bmatrix}.
\end{align*}
For $a\in(\mathbb{Z}/2\mathbb{Z})^{2n}$, the corresponding operator in $\mathcal{G}_n/\langle i I_2^{\otimes n}\rangle$ is denoted by $\sigma_a:=\sigma_{a_1a_{n+1}}\otimes\dotsm\otimes \sigma_{a_{n}a_{2n}}$.
\end{definition}

For row vectors $a, b \in (\mathbb{Z}/2\mathbb{Z})^{2n}$, $\sigma_a$ and $\sigma_b$ commute if and only if
\begin{align*}
a
\begin{bmatrix}
O_n&I_n\\
I_n&O_n
\end{bmatrix}
b^T
&=0
\end{align*}
where $O_n$ and $I_n$ are the $n\times n$ zero matrix and identity matrix on the field $\mathbb{F}_2$.
For a set $A=\{a_1,\dotsc,a_n\}\subseteq (\mathbb{Z}/2\mathbb{Z})^{2n}$ of size $n$,
a set of Pauli operators $\{\sigma_a\mid a\in A\}$ is a stabilizer generator if and only if
\begin{align*}
G
\begin{bmatrix}
O_n&I_n\\
I_n&O_n
\end{bmatrix}
G^T
&=O_n
\qquad
\text{and} \qquad \text{the rows of $G$ are linearly independent}
\end{align*}
where $G\in\mathbb{F}_2^{n\times 2n}$ denotes a matrix whose $j$-th row is $a_j$~\cite{PhysRevA.69.022316}.
Here, $G$ is called a \emph{generating matrix} of the stabilizer state.

In this study, stabilizer states with the same generating matrix are not distinguished since they are equivalent up to a Pauli operator.
\begin{restatable}{fact}{fasgn}
\label{fa:sgn}
For any two stabilizer states $\ket{\psi}$ and $\ket{\varphi}$ with the same generating matrices,
there exists a Pauli operator $g\in\mathcal{G}_n$ such that $\ket{\psi}$ is equal to $g\ket{\varphi}$ up to a constant factor.
\end{restatable}
The proof of Fact~\ref{fa:sgn} is shown in Appendix~\ref{apx:sgnL}.

\subsection{Clifford group}
\begin{definition}
Let $U(n)$ be a set of unitary matrices of size $n$. The \emph{Clifford group} is defined as
\begin{align*}
\mathcal{C}_n &:= \left\{ U\in U(2^n)\mid U\mathcal{G}_n U^\dagger = \mathcal{G}_n\right\}.
\end{align*}
\end{definition}

It is known that $\mathcal{C}_n$ is generated by the following three quantum gates
\begin{align*}
H&:=
\frac1{\sqrt{2}}
\begin{bmatrix}
1&1\\1&-1
\end{bmatrix}
,&
S&:=
\begin{bmatrix}
1&0\\0&i
\end{bmatrix}
,&
\CZ&:=
\begin{bmatrix}
1&0&0&0\\
0&1&0&0\\
0&0&1&0\\
0&0&0&-1\\
\end{bmatrix}.
\end{align*}

The set of \emph{local Clifford operators} is the subset of $\mathcal{C}_n$ defined as
\begin{align*}
\mathcal{C}_1^{\otimes n} &:= \left\{ C_1\otimes\dotsm\otimes C_n\mid C_i \in \mathcal{C}_1,\,\forall i\in\{1,\dotsc,n\} \right\}.
\end{align*}

A Clifford operator $C\in\mathcal{C}_n$ transforms a stabilizer state to another stabilizer state.
Applying a Clifford operator $C\in\mathcal{C}_N$ transforms a stabilizer $s\in\mathcal{G}_n$ to $CsC^\dagger$.
Indeed, $X$ and $Z$ are transformed as follows.
\begin{align}
HXH^\dagger &= Z,&
HZH^\dagger &= X,\nonumber\\
SXS^\dagger &= Y,&
SZS^\dagger &= Z\nonumber\\
\CZ (X\otimes I_2) \CZ^\dagger &= X\otimes Z,&
\CZ (I_2\otimes X) \CZ^\dagger &= Z\otimes X,\nonumber\\
\CZ (Z\otimes I_2) \CZ^\dagger &= Z\otimes I_2,&
\CZ (I_2\otimes Z) \CZ^\dagger &= I_2\otimes Z.
\label{eq:clifford}
\end{align}

Let $\sigma_{f_C(a)} := C\sigma_aC^\dagger$.
The following equality holds on $\mathcal{G}_n/\langle i I_2^{\otimes n}\rangle$
\begin{align*}
\sigma_{f_C(a+b)} = C \sigma_{a+b} C^\dagger = C \sigma_a \sigma_b C^\dagger = C \sigma_a C^\dagger C\sigma_b C^\dagger=\sigma_{f_C(a)}\sigma_{f_C(b)}=\sigma_{f_C(a)+f_C(b)}.
\end{align*}
Hence, the function $f_C\colon(\mathbb{Z}/2\mathbb{Z})^{2n}\to(\mathbb{Z}/2\mathbb{Z})^{2n}$ is linear so that there exists a unique $2n\times 2n$ $\mathbb{F}_2$-matrix $L_C$ satisfying $a L_C = f_C(a)$.
From~\eqref{eq:clifford}, we obtain the following $\mathbb{F}_2$-matrix representations,
\begin{align*}
L_H &= \begin{bmatrix}0&1\\1&0\end{bmatrix},&
L_S &= \begin{bmatrix}1&1\\0&1\end{bmatrix},&
L_\CZ &=
\begin{bmatrix}
1&0&0&1\\
0&1&1&0\\
0&0&1&0\\
0&0&0&1\\
\end{bmatrix}.
\end{align*}

Note that all six invertible $2\times 2$ matrices over $\mathbb{F}_2$ are generated by $L_H$ and $L_S$ as follows.
\begin{align}
L_I = L_H^2 = L_S^2 &= \begin{bmatrix}
1&0\\
0&1
\end{bmatrix},&
L_H &= \begin{bmatrix}
0&1\\
1&0
\end{bmatrix},&
L_SL_H &= \begin{bmatrix}
1&1\\
1&0
\end{bmatrix},&
L_S &= \begin{bmatrix}
1&1\\
0&1
\end{bmatrix},\nonumber\\
L_HL_SL_H = L_SL_HL_S &= \begin{bmatrix}
1&0\\
1&1
\end{bmatrix},&
L_HL_S &= \begin{bmatrix}
0&1\\
1&1
\end{bmatrix}.
\label{eq:1clifford}
\end{align}
Obviously, $L_{DC} = L_C L_D$ for any Clifford operations $C, D \in \mathcal{C}_n$.

In this study, Clifford operators with the same $\mathbb{F}_2$ matrix representation are not distinguished since they are equivalent up to a Pauli operator.
\begin{restatable}{fact}{faL}\label{fa:L}
Let $C, D\in\mathcal{C}_n$ be Clifford operators.
Then, $L_C=L_D$ if and only if there exists a Pauli operator $g\in\mathcal{G}_n$ such that $C$ is equal to $gD$ up to a constant factor.
\end{restatable}
The proof of Fact~\ref{fa:L} is shown in Appendix~\ref{apx:sgnL}.

\section{Graph state and combinatorial characterizations of local Clifford operations and measurements}\label{sec:graph0}
For a finite set $V$, let $\binom{V}{2}$ denote a family of subsets of $V$ of size exactly two.
For a finite set $V$, and $E\subseteq\binom{V}{2}$, a pair $(V,E)$ is called a \emph{graph}.
For a graph $G=(V, E)$, $V(G):=V$ and $E(G):=E$ are called \emph{vertex set} and \emph{edge set}, respectively.
The \emph{size of graph} $G$ is defined as $|V(G)|$.
For a graph $G$ and $v\in V(G)$,
the \emph{neighborhood of $v$ in $G$} is defined as $N_G(v):=\{w\in V\mid \{v,w\}\in E\}$.
\begin{definition}[Graph state]
For a graph $G=(V=\{1,2,\dotsc,n\},E)$, a \emph{graph state} $\ket{G}$ is a stabilizer state defined by stabilizer generators
\begin{align*}
X_v\otimes \bigotimes_{w\in N_G(v)} Z_w,\qquad v\in V.
\end{align*}
Here, for any $2\times 2$ matrix $U$, $U_v := g_1\otimes\dotsm\otimes g_n$ where $g_v = U$ and $g_w = I_2$ for $w\in V\setminus\{v\}$.
\end{definition}
The generating matrix of $\ket{G}$ is
$\begin{bmatrix}
I\mid A
\end{bmatrix}$
where $A\in\mathbb{F}_2^{n\times n}$ is the \emph{adjacency matrix} of $G$, i.e., $A_{v,w}=1$ if and only if $\{v,w\}\in E(G)$.

In Definition~\ref{def:comp}, the local Clifford operations are free to use.
Hence, the two stabilizer states that are equivalent up to the local Clifford operations can be identified in our analysis.
\begin{definition}[Local Clifford equivalence]
Stabilizer states $\ket{\psi}$ and $\ket{\varphi}$ are \emph{local-Clifford (LC) equivalent} if
there exists a local Clifford operator $C\in\mathcal{C}_1^{\otimes n}$ such that $C\ket{\psi} = \ket{\varphi}$.
\end{definition}

The following lemma implies that it is sufficient to prepare graph states for preparing general stabilizer states.
\begin{lemma}[\cite{PhysRevA.69.022316}]\label{lem:stabilizer}
For any stabilizer state $\ket{\psi}$, there exists a graph $G$ such that $\ket{\psi}$ and $\ket{G}$ are LC-equivalent.
\end{lemma}

For the combinatorial characterization of LC-equivalence of two graph states, the local complementation is defined as follows.
\begin{definition}[Local complementation]
The \emph{local complementation} is an operation on a graph.
For a graph $G=(V,E)$, and a vertex $v\in V$,
the local complementation $\tau_v(G)=(V,E')$ is defined as
\begin{align*}
E' &= E \bigtriangleup \binom{N_G(v)}{2}
\end{align*}
where $\bigtriangleup$ denotes the symmetric difference.
In other words, the local complementation at $v$ leaves all adjacencies outside $N_G(v)$ unchanged and complements the subgraph induced by $N_G(v)$.

Two graphs $G$ and $H$ are said to be \emph{local-complement equivalent} if $G$ and $H$ can be transformed to each other by a sequence of local complementations.
\end{definition}
Figure~\ref{fig:lc} describes an example of the local complementation.

\begin{lemma}[Local Clifford equivalence of graph states~\cite{PhysRevA.69.022316}]~\label{lem:lc}
For any graphs $G$ and $H$, $\ket{G}$ and $\ket{H}$ are local-Clifford equivalent if and only if
$G$ and $H$ are local-complement equivalent.
\end{lemma}

Lemma~\ref{lem:lc} provides a combinatorial characterization of local Clifford operations on graph states.
Next, we consider a combinatorial characterization of measurements.
In quantum algorithms in Definition~\ref{def:comp}, general stabilizer states, which are not necessarily graph states, may be measured in the computational basis.
Since any stabilizer state is LC-equivalent to some graph state,
it is sufficient to consider Pauli measurements for graph states.

\begin{lemma}[Pauli measurements for graph states~\cite{PhysRevA.69.062311}]~\label{lem:paulim}
For any graph $G$, let $\ket{G^+}$ and $\ket{G^-}$ be two possible stabilizer states obtained by a Pauli measurement, i.e., in the eigenbasis of $X$, $Y$ or $Z$, of a qubit in $\ket{G}$.
Then, $\ket{G^+}$ and $\ket{G^-}$ are LC-equivalent to some common graph state $\ket{H}$.
Here, $H$ is obtained from $G$ by some local complementations and deletion of the vertex that corresponds to the measured qubit.
\end{lemma}

Lemma~\ref{lem:paulim} implies that a Pauli measurement and an adaptive local Clifford operation for a graph state $\ket{G}$ give some graph state $\ket{H}$ deterministically.
Furthermore, the graph $H$ is obtained by local complementation and a vertex deletion for $G$.
This observation implies the following lemma.

\begin{definition}[Vertex minor~\cite{BOUCHET198858,dahlberg2018}]
A graph $H$ is a \emph{vertex minor} of graph $G$ if $H$ is obtained from $G$ by local complementations and vertex deletions.
\end{definition}

\begin{lemma}[\cite{dahlberg2018,PhysRevA.69.062311,BOUCHET198858}]\label{lem:vm}
For any graphs $G$ and $H$, $\ket{H}$ is obtained from $\ket{G}$ with probability 1 by a local Clifford operation followed by measurement of qubits in the computational basis and an adaptive local Clifford operation if and only if
$H$ is a vertex-minor of $G$.
\end{lemma}
Lemmas~\ref{lem:vm} provides combinatorial characterizations of local Clifford operations and measurements in the computational basis.
Dahlberg, Helsen, and Wehner showed that it is NP-hard to decide whether $G$ is a vertex-minor of $H$ for given graphs $G$ and $H$~\cite{Dahlberg_2020}.

\begin{figure}[t]
\centering
\begin{subfigure}{0.4\hsize}
\centering
\begin{tikzpicture}[rotate=90]
\node[circle,draw,very thick, fill=red] (a) at (0:1) {};
\node[circle,draw,very thick] (b) at (72:1) {};
\node[circle,draw,very thick] (c) at (144:1) {};
\node[circle,draw,very thick] (d) at (216:1) {};
\node[circle,draw,very thick] (e) at (288:1) {};
\node[circle,draw,very thick] (f) at (0,0) {};
\draw[very thick] (a) -- (b) -- (c) -- (d) -- (e) -- (a);
\draw[very thick] (f) -- (a);
\draw[very thick] (f) -- (b);
\draw[very thick] (f) -- (c);
\draw[very thick] (f) -- (d);
\draw[very thick] (f) -- (e);
\end{tikzpicture}
\end{subfigure}
\begin{subfigure}{0.4\hsize}
\centering
\begin{tikzpicture}[rotate=90]
\node[circle,draw,very thick, fill=red] (a) at (0:1) {};
\node[circle,draw,very thick] (b) at (72:1) {};
\node[circle,draw,very thick] (c) at (144:1) {};
\node[circle,draw,very thick] (d) at (216:1) {};
\node[circle,draw,very thick] (e) at (288:1) {};
\node[circle,draw,very thick] (f) at (0,0) {};
\draw[very thick] (a) -- (b) -- (c) -- (d) -- (e) -- (a);
\draw[very thick] (f) -- (a);
\draw[very thick] (f) -- (c);
\draw[very thick] (f) -- (d);
\draw[very thick] (b) -- (e);
\end{tikzpicture}
\end{subfigure}
\caption{Local complementation at the red top vertex transforms the graph on the left into the graph on the right, and vice versa.}\label{fig:lc}
\end{figure}

\section{Combinatorial characterization of two-qubit Clifford operations: Proof of Theorem~\ref{thm:c1}}\label{sec:graph1}

In this section, we consider a combinatorial characterization of two-qubit Clifford operations.
For the combinatorial interpretation of the \CZ operation, we define elementary edge-complementations as follows.
The three transformations below were derived for graphical descriptions of \CZ operations on stabilizer states in~\cite{PhysRevA.77.042307}.

\begin{definition}[{Elementary edge-complementations}]\label{def:cost1}
For a graph $G=(V,E)$, an \emph{elementary edge-complementation} is defined as any one of the following:
\begin{enumerate}
\item For $\{v,w\}\in \binom{V}{2}$, complement an edge $\{v,w\}$, i.e., add the edge if it does not exist, remove it if it does.
\item For $\{v,w\}\in \binom{V}{2}$, complement all edges between $v$ and $N_G(w)\setminus\{v\}$.
\item For $\{v,w\}\in \binom{V}{2}\setminus E$, complement all edges in $\bigl\{\{x,y\}\in \binom{V}{2}\mid x\in N_G(v),\, y\in N_G(w),\, \{x,y\}\nsubseteq N_G(v)\cap N_G(w)\bigr\}$.
\end{enumerate}
The second operation is called \emph{the edge complementation between $v$ and $N_G(w)$}.
The third operation is called \emph{the edge complementation between $N_G(v)$ and $N_G(w)$}.
\end{definition}

It is easy to see that the above operations correspond to some two-qubit Clifford operations.
The first operation corresponds to the \CZ operation.
Let $\zeta^{(1)}_{v,w}$ be a mapping on graphs corresponding to the edge complementation for $\{v,w\}\in \binom{V}{2}$.
The second operation is equivalent to
\begin{align*}
\zeta^{(1)}_{v,w}\circ\tau_w\circ\zeta^{(1)}_{v,w}\circ\tau_w &=
\tau_w\circ\zeta^{(1)}_{v,w}\circ\tau_w\circ\zeta^{(1)}_{v,w}.
\end{align*}
The third operation for $\{v,w\}\in \binom{V}{2}\setminus E$ is equivalent to
\begin{align*}
\zeta^{(1)}_{v,w}\circ\tau_v\circ\tau_w\circ\tau_v\circ\zeta^{(1)}_{v,w}
&=
\zeta^{(1)}_{v,w}\circ\tau_w\circ\tau_v\circ\tau_w\circ\zeta^{(1)}_{v,w}
\end{align*}
and swapping the labels of the vertices $v$ and $w$.
Note that $\tau_v\circ\tau_w\circ\tau_v = \tau_w\circ\tau_v\circ\tau_w$ for $\{v,w\}\in E$ is called \emph{the pivot operation}, which is the edge complementation between $N_G(v)$ and $N_G(w)$ and the swapping of the labels of $v$ and $w$~\cite{BOUCHET1994107,OUM200579}.
Indeed, the three operations in Definition~\ref{def:cost1} explain all two-qubit Clifford operations.

We use a detailed version of Lemma~\ref{lem:stabilizer} for the proof of Theorem~\ref{thm:c1}.
\begin{lemma}[\cite{PhysRevA.77.042307}]\label{lem:stabilizer1}
For any stabilizer state $\ket{\psi}$, 
there exist a graph $G=(V,E)$ and $A, B\subseteq V$ such that
$A\cap B = \varnothing$,
$G$ has no edge between vertices in $B$, and
$\ket{G}$ is equal to
$\bigotimes_{v\in A} S_v \bigotimes_{w\in B} H_w \ket{\psi}$
up to a Pauli operator.
\end{lemma}

The following lemma was essentially proved in~\cite{PhysRevA.77.042307}.
\begin{lemma}[{\cite{PhysRevA.77.042307}}]\label{lem:hczh}
For any graph $G$, the following equalities hold up to a Pauli operator:
\begin{align*}
H_b\CZ_{a,b}H_b \ket{G} &= \ket{\zeta_{a,b}^{(2)}(G)}\qquad \forall \{a,b\}\in\binom{V}{2}\\
(H_a\otimes H_b)\CZ_{a,b}(H_a\otimes H_b) \ket{G} &= \ket{\zeta_{a,b}^{(3)}(G)}\qquad \forall \{a,b\}\in\binom{V}{2}\setminus E(G)
\end{align*}
where $\zeta_{a,b}^{(2)}$ denotes the edge complementation between $a$ and $N_G(b)$, and $\zeta_{a,b}^{(3)}$ denotes the edge complementation between $N_G(a)$ and $N_G(b)$.
\end{lemma}
From Lemmas~\ref{lem:stabilizer1} and~\ref{lem:hczh}, we obtain the combinatorial characterization of quantum algorithms with \CZ-complexity one.
\begin{theorem}\label{thm:cost1}
For any graphs $G$ and $H$, graph states $\ket{G}$ and $\ket{H}$ can be transformed to each other by an arbitrary number of local Clifford operations and one \CZ operation if and only if $G$ and $H$ can be transformed to each other by an arbitrary number of local complementations and a single elementary edge-complementation.
\end{theorem}
\begin{proof}
From Lemmas~\ref{lem:lc} and~\ref{lem:hczh}, the direction $\Leftarrow$ is obvious.
We will prove the other direction.
Let $\ket{\psi}$ be a stabilizer state just before the \CZ operation is performed in the transformation from $\ket{G}$ to $\ket{H}$.
From Lemma~\ref{lem:stabilizer1}, there exist a graph $G'$ and $A, B\subseteq V$ such that
$A\cap B=\varnothing$, $G'$ has no edge between vertices in $B$ and $\bigotimes_{v\in A} S_v \bigotimes_{w\in B} H_w \ket{\psi}$ is equal to $\ket{G'}$ up to a Pauli operator.
Then, some CZ operation is applied to a pair $\{a,b\}\in\binom{V}{2}$ of qubits of $\ket{\psi}$.
Assume $a$ and $b$ are both not in $B$.
Then, up to a Pauli operator,
\begin{align*}
\left(\bigotimes_{v\in A} S_v\bigotimes_{w\in B}H_w\right) \CZ_{a,b} \ket{\psi}
=
\CZ_{a,b} \left(\bigotimes_{v\in A} S_v\bigotimes_{w\in B}H_w\right)  \ket{\psi}
=
\CZ_{a,b} \ket{G'}
=
\ket{\zeta^{(1)}_{a,b}(G')}
\end{align*}
since the \CZ gate and $S$ gate commute.
Assume $a$ is not in $B$ and $b$ is in $B$.
Then, up to a Pauli operator,
\begin{align*}
\left(\bigotimes_{v\in A} S_v\bigotimes_{w\in B}H_w\right) \CZ_{a,b} \ket{\psi}
&=
\left(\bigotimes_{v\in A} S_v\bigotimes_{w\in B}H_w\right) \CZ_{a,b} \left(\bigotimes_{v\in A} S^\dagger_v\bigotimes_{w\in B}H_w\right)\ket{G'}\\
&=
H_b \CZ_{a,b} H_b \ket{G'}
=
\ket{\zeta^{(2)}_{a,b}(G')}.
\end{align*}
Finally, assume $a$ and $b$ are both in $B$.
Then, up to a Pauli operator,
\begin{align*}
\left(\bigotimes_{v\in A} S_v\bigotimes_{w\in B}H_w\right) \CZ_{a,b} \ket{\psi}
&=
\left(\bigotimes_{v\in A} S_v\bigotimes_{w\in B}H_w\right) \CZ_{a,b} \left(\bigotimes_{v\in A} S^\dagger_v\bigotimes_{w\in B}H_w\right)\ket{G'}\\
&=
(H_a \otimes H_b) \CZ_{a,b} (H_a \otimes H_b) \ket{G'}
=
\ket{\zeta^{(3)}_{a,b}(G')}.
\end{align*}

Since $\ket{G}$ and $\ket{G'}$ are LC-equivalent, $G$ and $G'$ are local-complement equivalent.
Then, the transformations of graph $G\rightarrow G'\rightarrow \zeta_{a,b}^{(i)}(G') \rightarrow H$
for some $i\in\{1,2,3\}$ are realized by an arbitrary number of local complementations and a single elementary edge-complementation.
\end{proof}

From Lemmas~\ref{lem:lc},~\ref{lem:vm} and Theorem~\ref{thm:cost1}, the CZ-complexity of graph state is now characterized in a combinatorial way.

\begin{definition}[The edge-complementation-complexity of graphs]\label{def:compG}
Let $\mathcal{A}$ be an algorithm acting on a graph that consists of the following operations.
\begin{enumerate}
\item Deletion of a vertex.
\item Local complementation.
\item Elementary edge-complementation.
\end{enumerate}
The \emph{edge-complementation-complexity of $\mathcal{A}$} is defined as the number of elementary edge-complementations in $\mathcal{A}$.
The \emph{edge-complementation-complexity (EC-complexity) $\NCO(G)$ of a graph $G$} is defined as the minimum edge-complementation-complexity of $\mathcal{A}$
that generates $G$ with $n$ vertices from the empty graph with $n+s$ vertices for some $s\ge 0$.
\end{definition}

\begin{theorem}[Equivalent to Theorem~\ref{thm:c1}]
For any graph $G$, $\NCZ(\ket{G})= \NCO(G)$.
\end{theorem}

In the rest of this paper, we mainly consider $\NCO(G)$ rather than $\NCZ(\ket{G})$ since it is easier to analyze.

\section{Rank-width}\label{sec:rank}
The rank-width is a complexity measure of a graph~\cite{OUM2006514,OUM201715}.
The rank-width is defined using a cut-rank.
\begin{definition}[Cut-rank]
For a graph $G=(V,E)$ and $S\subseteq V$, the \emph{bipartite adjacency matrix} $A_{S, V\setminus S}\in\mathbb{F}_2^{|S|\times|V\setminus S|}$ is a matrix whose rows and columns are indexed by $s\in S$ and $t\in V\setminus S$, respectively. An $(s,t)$-entry of $A_{S,V\setminus S}$ is 1 if and only if $\{s,t\}\in E$.
The \emph{cut-rank} of $S\subseteq V$ for $G$ is defined as $\cutrank_G(S):=\mathrm{rank}_{\mathbb{F}_2}(A_{S,V\setminus S})$.
The cut-rank of the empty set is defined as zero.
\end{definition}

Obviously, 
the cut-rank is symmetric, i.e., $\cutrank_G(V\setminus S) = \cutrank_G(S)$ for any $S\subseteq V$.
The cut-rank plays an important role in this work since it is invariant under the local complementations.
\begin{lemma}[\cite{bouchet1989connectivity,OUM200579}]\label{lem:cutrank}
For any $G=(V,E)$, $S\subseteq V$ and $v\in V$, $\cutrank_G(S) = \cutrank_{\tau_v(G)}(S)$.
\end{lemma}
\begin{proof}
Let $A_{S,V\setminus S}\in\mathbb{F}_2^{|S|\times|V\setminus S|}$ be the bipartite adjacency matrix of $G$.
Without loss of generality, we can assume $v\in S$.
Then, the bipartite adjacency matrix for $\tau_v(G)$ is obtained from $A_{S,V\setminus S}$ by adding the $v$-th row to all rows that correspond to a neighborhood of $v$.
This transform does not change the rank of the bipartite adjacency matrix.
\end{proof}

The rank-width is defined as follows.
\begin{definition}[Rank-decomposition and rank-width~\cite{OUM2006514,OUM201715}]
Let $G=(V,E)$ be a graph. 
A \emph{subcubic tree} is a tree whose maximum degree is at most three.
The \emph{rank-decomposition} of $G$ is a pair $(T, L)$ of a subcubic tree $T$ with at least two vertices and a bijection $L$ from $V$ to a set of all leaves of $T$. 
For a rank-decomposition $(T=(W, F), L)$ and an edge $e\in F$, $T-e$ determines a partition of leaves of $T$, equivalently a partition of $V$ determined by $L$.
The \emph{width of $e\in F$} is $\cutrank_G(S)$ where $S, V\setminus S$ is the partition of $V$ determined by $e$. 
The \emph{width of the rank-decomposition $T$} is the maximum of the width among all edges in $T$.
The \emph{rank-width} $\rw(G)$ of $G$ is the minimum width over all rank-decompositions of $G$.
When $G$ has at most one vertices, $G$ has no rank-decomposition. In this case, the rank-width of $G$ is defined as zero.
\end{definition}

Obviously, the rank-width of an $n$-vertex graph is at most $\lceil n/3\rceil$ since any rank-decomposition containing a node whose removal induces a balanced tripartition of the leaves has width at most $\lceil n/3\rceil$.

\begin{fact}\label{fact:rwm}
For a graph $G=(V,E)$, let $G[W]$ denote the induced subgraph in $G$ by $W\subseteq V$.
In other words, $G[W] := (W, E \cap \binom{W}{2})$.
The cut-rank and rank-width are monotone with respect to induced subgraphs, i.e.,
\begin{itemize}
\item For any $S\subseteq W\subseteq V$, $\cutrank_{G[W]}(S) \le \cutrank_G(S)$.
\item For any $W\subseteq V$, $\rw(G[W])\le\rw(G)$.
\end{itemize}
\end{fact}

The rank-width can be computed in time $\mathrm{poly}(n) 2^n$~\cite{OUM2009745}.
There also exist fixed parameter-tractable algorithms computing the rank-width~\cite{COURCELLE200791,doi:10.1137/22M153937X,10.1145/3618260.3649732,korhonen2026branch}.

Here, we provide the proof of Lemma~\ref{lem:inf} using the information-theoretic counting argument.
\leminf*
\begin{proof}
Any bipartite graph with $r$ vertices at one side and $n-r$ vertices at the other side has rank-width at most $r$ since any bipartition has cut-rank at most $r$.
Hence, the number of graphs with $n$ vertices and rank-width at most $r$ is at least $2^{r(n-r)}$.
All graph states have different generating matrices.
Since quantum algorithms in Definition~\ref{def:comp} have to generate $\ket{G}$ with probability 1, the quantum algorithm generates $\ket{G}$ when the all measurement outcomes for the $s$ working qubits are zero.
In this case, we can consider a corresponding Clifford circuit with the projectors $\ket{0}\bra{0}$ for the $s$ working qubits.
The output of the Clifford circuit with the projectors must be proportional to $\ket{G}$.
In the following, we count the number of Clifford circuits with the projectors.
There are $324$ two-qubit Clifford operators acting on the generating matrices from Figure~\ref{fig:2clifford}~(b) in Appendix~\ref{apx:cz}.
Non-working qubits that do not appear in any two-qubit operation must be $\ket{+}$.
Then, the number of Clifford circuits including at most $t$ two-qubit Clifford gates is at most
\begin{align*}
\left(1+324\binom{n+s}{2}\right)^t.
\end{align*}
We can assume without loss of generality that all working qubits ``connect'' to some non-working qubit by two-qubit Clifford operations.
This observation implies $s\le t\le \binom{n}{2}$.
If $t$ two-qubit Clifford gates are sufficient for all graph states with $n$ vertices and rank-width $r$,
\begin{align*}
\left(1+324\binom{n+n^2}{2}\right)^{t} &\ge 2^{r(n-r)} \ge 2^{r\lfloor\frac23 n\rfloor}
\end{align*}
since $r\le \lceil n/3\rceil$.
Then, we obtain
\begin{align*}
t \ge \frac{r}{\log_2\left(1+324\binom{n+n^2}{2}\right)}\left\lfloor\frac23n\right\rfloor.
\end{align*}
\end{proof}

Here, we also provide the proof of Lemma~\ref{lem:tlower} using the rank-width of random sparse graph analyzed in~\cite{https://doi.org/10.1002/jgt.20620}.
\lemtlower*
\begin{proof}[Proof of Lemma~\ref{lem:tlower}]

Let $G(n, p)$ denote the Erd\H{o}s--R\'enyi random graph.
The rank-width of $G(n, c/n)$ with constant $c>1$ was analyzed in~\cite{https://doi.org/10.1002/jgt.20620}.
\begin{lemma}[Rank-width of random sparse graph~\cite{https://doi.org/10.1002/jgt.20620}]
\label{lem:random}
For any constant $c>1$, there exist constants $d>0$ and $e>0$ such that $G(n,c/n)$ contains a connected subgraph $H$ satisfying $|V(H)|\ge dn$ and $\rw(H)\ge en$ asymptotically almost surely.
\end{lemma}
Here, $G(n, c/n)$ contains at most $(c/2)n-1$ edges with probability asymptotically $1/2$.
Hence, for $c > 1$,
\begin{align}
\NCO(H)\le (c/2)n - 1 &= dn - 1 + (c/2-d) n
\le |V(H)| - 1 + [(c/2-d)/e] \rw(H).
\label{eq:H}
\end{align}
holds with probability asymptotically at least $1/2$.

Then, fix $c>1$ arbitrarily.
For any $r'$, let $n'$ be a minimum integer such that there exists a connected graph $H$ satisfying~\eqref{eq:H}, $|V(H)|=n'$ and $\rw(H)\ge r'$.
Note that such graph $H$ always exists since asymptotically at least half of $H$ in Lemma~\ref{lem:random} satisfy~\eqref{eq:H} and have linear rank-width.
Adding a new vertex of degree one to $H$ gives a new graph $H'$ such that $|V(H')|=|V(H)|+1$, $\NCO(H')\le\NCO(H)+1$ and $\rw(H')=\rw(H)$.
Hence,~\eqref{eq:H} still holds for $H'$.
This argument implies that for any $r'$, there exist $n'$ and $r\ge r'$ such that for any $n\ge n'$ there exists a graph $H$ with $n$ vertices and rank-width $r$ satisfying $\NCO(H)\le n - 1 + [(c/2-d)/e] r$.
\end{proof}


Determining the minimum value of the constant $(c/2-d)/e$ is of interest, particularly whether it can attain the lower bound of 1 given in Theorem~\ref{thm:lower}.

\section{Lower bound: Proof of Theorem~\ref{thm:lower}}\label{sec:lower}
In this section, we prove Theorem~\ref{thm:lower}.
\thmlower*
The following lemma states that the cut-rank is unchanged by an elementary edge-complementation that does not cross the cut, and is changed by at most one by an elementary edge-complementation that crosses the cut.
\begin{lemma}\label{lem:cost1}
For a graph $G=(V,E)$ and a pair $\{v,w\}\in \binom{V}{2}$ of vertices,
let $G'$ be a graph obtained by an elementary edge-complementation for $\{v,w\}$ from $G$.
Then,
\begin{itemize}
\item For any $W\subseteq V$, $\cutrank_G(W)=\cutrank_{G'}(W)$ when $(v\in W\land w\in W)\lor(v\notin W\land w\notin W)$.
\item For any $W\subseteq V$, $|\cutrank_G(W)-\cutrank_{G'}(W)|\le 1$ when $(v\in W\land w\notin W)\lor(v\notin W\land w\in W)$.
\end{itemize}
\end{lemma}
\begin{proof}
First, let us consider the case $v\in W\land w\in W$.
Since the bipartite adjacency matrix $A_{W,V\setminus W}$ does not change by the edge complementation between $v$ and $w$, its rank does not change.
For the edge complementation between $v$ and $N_G(w)$, the new bipartite adjacency matrix is obtained by adding the $w$-th row to the $v$-th row for $A_{W,V\setminus W}$.
Hence, the rank does not change.
For the edge complementation between $N_G(v)$ and $N_G(w)$ for vertices $v$ and $w$ non-adjacent in $G$, the new bipartite adjacency matrix is obtained by adding the $v$-th row to rows in $N_G(w)\cap W$ and adding the $w$-th row to rows in $N_G(v)\cap W$.
Hence, the rank does not change.
Similarly, for the case $v\notin W\land w\notin W$, the rank of the bipartite adjacency matrix does not change.

Next, let us consider the case $v\in W\land w\notin W$.
For the edge complementation between $v$ and $w$, the new bipartite adjacency matrix is obtained by flipping the $(v,w)$-entry.
Hence, the rank changes by at most one.
For the edge complementation between $v$ and $N_G(w)$, the new bipartite adjacency matrix is obtained by flipping $(v,u)$-entry for $u\in N_G(w)\setminus W$.
Hence, the rank changes by at most one.
For the edge complementation between $N_G(v)$ and $N_G(w)$ for vertices $v$ and $w$ non-adjacent in $G$, the new bipartite adjacency matrix is obtained by adding the $v$-th row to rows in $N_G(w)\cap W$ and adding a row vector $1_{N_G(w)\setminus W}$ to rows in $N_G(v)\cap W$ where $1_{N_G(w)\setminus W}$ is a row vector whose $u\in V\setminus W$-th element is 1 if and only if $u\in N_G(w)\setminus W$.
Hence, the rank changes by at most one.
Similarly, for the case $v\notin W\land w\in W$, the rank of the bipartite adjacency matrix changes by at most one.
\end{proof}

\begin{theorem}[Equivalent to Theorem~\ref{thm:lower}]\label{thm:lowerEC}
For any connected graph $G$ with $n\ge 2$ vertices and rank-width at least $r$, $\NCO(G)\ge n+r-2$.
\end{theorem}
\begin{proof}
Let $G=(V,E)$ be a connected graph with at least two vertices.
Then, $\rw(G)\ge 1$ so that we assume $r\ge 1$ in the following.
From the connectivity of $G$, $\NCO(G)\ge n-1$.
Hence, Theorem holds for $r=1$.

In the following, we assume $r\ge 2$ and show that
if $\NCO(G)\le n+r-3$, then, $\rw(G)\le r-1$.
Let $S$ be the set of vertices of size $s$ that are deleted in an algorithm in Definition~\ref{def:compG}.
Without loss of generality, we can assume that the vertices are deleted at the end of the algorithm.
Let $G'=(V\cup S, E')$ be a graph just before the vertex deletion.
Then, $G'[V] = G$.
From Fact~\ref{fact:rwm}, it is sufficient to show $\rw(G')\le r-1$.
In the following, we will show $\rw(G')\le r-1$.
Let $k:=\NCO(G)$.
Let $e_1,e_2,\dotsc,e_k$ be a sequence of pairs of vertices of $G'$ for which the elementary edge-complementations are applied in this order.
Furthermore, let
\begin{align*}
F:=\Bigl\{ e_i=\{v,w\} &\mid
i\in\{1,\dotsc,k\}\quad
\text{$v$ and $w$ belong to different connected components}\\
&\qquad \text{in a graph $(V\cup S, \{e_1,\dotsc,e_{i-1}\})$}
\Bigr\}.
\end{align*}
Let $T:=(V\cup S, F)$.
From the definition of $F$, $T$ is a forest.
From the connectivity of $G$, all vertices in $V$ must belong to a common connected component of $T$.
Since all other vertices that do not belong to the connected component including $V$ does not affect the generation of the graph $G$, we can assume that there does not exist such vertices without loss of generality.
In the following, we assume that $|F|=n+s-1$, and $T$ is a tree.

\begin{figure}[t]
  \centering
  \begin{tikzpicture}[inner sep = 0, minimum size = 20pt, very thick, scale=1.5]
    \node[circle,radius=100pt,draw] (c) at (4,0) {$v$};
    \draw (c) -- node[above=0.0em,at end] {$e_{v_1}$} (5.5,2);
    \draw (c) -- node[above=0.0em,at end] {$e_{v_2}$} (5.5,1);
    \draw (c) -- node[above=0.0em,at end] {$e_{v_3}$} (5.5,0);
    \node at (6.5,0) {$\bm{\longmapsto}$};
    \node[circle,draw] (c) at (7.7,0) {$l_v$};
    \node[circle,draw] (d) at (9.2,0) {$b_{v,3}$};
    \node[circle,draw] (e) at (9.2,1) {$b_{v,2}$};
    \node[circle,draw] (f) at (9.2,2) {$b_{v,1}$};
    \draw (c) -- (d) -- (e) -- (f);
    \draw (f) -- (10.7,2);
    \draw (e) -- (10.7,1);
    \draw (d) -- (10.7,0);
  \end{tikzpicture}
  \caption{Construction of rank-decomposition on the basis of the tree $T$.
  Left: A vertex $v$ in $T$.
  Right: A leaf of the rank-decomposition corresponding to $v$.}
  \label{fig:gadget}
\end{figure}

We use $T$ to construct a rank-decomposition of $G'$ of width at most $r-1$ as follows.
For each $v\in V\cup S$, let $e_{v_1},\dotsc,e_{v_{d_v}} \in F$ be edges of $T$ incident to $v$ where $v_1< v_2< \dotsb< v_{d_v}$.
A gadget $H_v$ corresponding to the vertex $v$ is defined as
\begin{align*}
V(H_v) &:= \{l_v\}\cup\{b_{v,i}\mid i=1,\dotsc,d_v\}\\
E(H_v) &:= \{\{l_v,b_{v,d_v}\}\}\cup \{\{b_{v,i},b_{v,i+1}\}\mid i= 1,\dotsc,d_v-1\}
\end{align*}
Figure~\ref{fig:gadget} shows the correspondence between $v$ and $H_v$.
We now define the rank-decomposition $H$ of $G'$ as
\begin{align*}
V(H) &:= \bigcup_{v\in V\cup S} V(H_v)\\
E(H) &:= \bigcup_{v\in V\cup S} E(H_v) \cup \bigcup_{\{v,w\}\in F} \bigl\{\{b_{v,i},b_{w,j}\}\mid v_i=w_j, i\in\{1,\dotsc,d_v\}, j\in\{1,\dotsc,d_w\}\bigr\}.
\end{align*}
Then, it is obvious that $H$ is a subcubic tree and the set of leaves of $H$ is $\{l_v\mid v\in V\cup S\}$. 
Hence, $(H, L\colon v \mapsto l_v)$ is a rank-decomposition of $G'$.

We will show that the width of the rank-decomposition $(H, L)$ is at most $r-1$.
Each edge $e\in E(H)$ defines a partition $P_e:=\{V_e, (V\cup S)\setminus V_e\}$ of vertices of $G'$.
In the following, we will show $\cutrank_{G'}(V_e)\le r-1$ for all $e\in E(H)$.

If $e=\{l_v, b_{v,d_v}\}$ for some $v\in V\cup S$, then $P_e=\{\{v\}, (V\cup S)\setminus\{v\}\}$. This implies $\cutrank_{G'}(V_e)=1$.

If $e=\{b_{v,i},b_{w,j}\}$ for some $v\ne w$, edges in $F\setminus\{v,w\}$ do not cross the cut $V_e$.
The number of elementary edge-complementations that cross the cut $V_e$ is at most $k-|F|+1=k-n-s+2\le (n+r-3)-n-s+2 = r-s-1\le r-1$.
From Lemmas~\ref{lem:cutrank} and \ref{lem:cost1}, $\cutrank_{G'}(V_e)\le r-1$.

Finally, assume $e=\{b_{v,i},b_{v,i+1}\}$, for some $v\in V\cup S$.
Without loss of generality, we can assume $v\in V_e$.
Before the elementary edge-complementation corresponding to $e_{v_{i+1}}$ is applied,
$v$ is the unique vertex in $V_e$ that may connect to $(V\cup S)\setminus V_e$.
This means that the cut-rank of $V_e$ is at most one at this time.
After $e_{v_{i+1}}$, the number of elementary edge-complementations that cross the cut $V_e$
is at most $k-|F|\le r-2$.
From Lemmas~\ref{lem:cutrank} and~\ref{lem:cost1}, $\cutrank_{G'}(V_e)\le r-1$.

From these observations, $\cutrank_{G'}(V_e)\le r-1$ for all $e\in E(H)$, which implies that the width of the rank-decomposition $H$ is at most $r-1$.
\end{proof}

\section{Upper bound: Proof of Theorem~\ref{thm:upper}}\label{sec:upper}

In this section, we present algorithms that generate $G$ with EC-complexity $O(rn)$ and prove the following simplified version of Theorem~\ref{thm:upper}.
\begin{theorem}[Simplified version of Theorem~\ref{thm:upper}]\label{thm:simpu}
Let $G$ be a graph with $n$ vertices and rank-width at most $r\ge 1$.
Then,
\begin{align*}
\NCZ(\ket{G}) &\le
\begin{cases}
\frac{5r^2-1}{4r} n + O(1)&\text{if $r$ is odd}\\
\frac{5r}{4} n + O(1)&\text{otherwise.}
\end{cases}
\end{align*}
\end{theorem}
The constant terms $O(1)$ with respect to $n$ are calculated in Appendix~\ref{apx:ind}, which
proves Theorem~\ref{thm:upper}.

The following is the key lemma for the algorithms in this section.

\begin{lemma}\label{lem:grow}
For any graph $G$ and $S\subseteq V(G)$, if $\cutrank_G(S)<|S|$, i.e., the rows of $A_{S,V\setminus S}$ are linearly dependent, there exists $v\in S$ such that $\NCO(G)\le\NCO(G-v)+|S|-1$.
\end{lemma}
\begin{proof}
From the linear dependence of rows of $A_{S,V\setminus S}$,
there exists $v\in S$ and $T\subseteq S\setminus\{v\}$ such that the $v$-th row of $A_{S,V\setminus S}$ is a sum of rows of $A_{S,V\setminus S}$ in $T$.
Let $A$ be the adjacency matrix of $G$.
Let $A_v$ be the $v$-th row of $A$.
Let $A_T:=\sum_{w\in T} A_w\bmod 2$.
Then, for any $u\in V\setminus S$, an $u$-th element of $A_v$ and $A_T$ are equal.
Let
\begin{align*}
U := \{ u\in T\mid \text{$u$-th elements of $A_v$ and $A_T$ are different}\}.
\end{align*}
Then, we can generate $G$ from $G-v$ by the following algorithm.
First, for $w\in U$, we apply $\tau_w\circ\zeta_{v,w}^{(1)}\circ\tau_w$,
and for $w\in T\setminus U$, we apply $\zeta_{v,w}^{(2)}$.
Then, the neighborhood of $v$ is correctly generated for all $(V\setminus S)\cup T$.
Finally, we correct the neighborhood of $v$ for $S\setminus (T\cup\{v\})$ by $\zeta^{(1)}_{v,w}$ if needed.
The EC-complexity of this algorithm is at most $|S|-1$.
\end{proof}


Algorithm~\ref{alg:gengen} provides the algorithmic framework for generating graphs that we will consider in the rest of this section.
The EC-complexity of Algorithm~\ref{alg:gengen} is determined by the size of $S$ in Line 4.
\begin{definition}\label{def:dep}
For a graph $G$,
a subset $S\subseteq V(G)$ of vertices satisfying $\cutrank_G(S) < |S|$ is called a \emph{dependent set}.
\end{definition}
Indeed, $S$ in Definition~\ref{def:dep} is a dependent set of the independence system in Appendix~\ref{apx:ind}.
In the following, we consider upper bounds on the size of the minimum dependent sets.

\begin{algorithm}[t!]
\caption{A framework of algorithms to generate a graph.}\label{alg:gengen}
\Input{A graph $G$}
\Output{A graph $G$ generated by the operations in Definition~\ref{def:compG}.}
\Procedure{\emph{\Call{GenGraph}{$G$}}}{
\If{$G$ has at most two vertices}{
\Return \text{a graph $G$ generated in a trivial way}\;
}
Find a small set $S\subseteq V(G)$ with $\cutrank_G(S) < |S|$\;
Find $v\in S$ such that the $v$-th row of $A_{S,V\setminus S}$ is spanned by other rows\;
$G'\gets \Call{GenGraph}{G[V(G)\setminus\{v\}]}$\;
Add $v$ to $G'$ with EC-complexity at most $|S|-1$ by the algorithm in Lemma~\ref{lem:grow}\;
\Return $G'$\;
}
\end{algorithm}


The following lemma provides an upper bound on the size of the minimum dependent set when a graph $G$ has rank-width $r$.
\begin{lemma}\label{lem:part}
For any positive integer $r$ and a subcubic tree $T$ with $n\ge 3r$ leaves, there exists an edge $e \in E(T)$ such that one of the two trees obtained by removing $e$ from $T$ contains $k\in\{r+1,r+2,\dotsc,2r\}$ leaves of $T$.
\end{lemma}
\begin{proof}
We will show that the edge $e \in E(T)$ selected by Algorithm~\ref{alg:treer} satisfies the conditions of the lemma.

We first assume $n\ge 3r+1$.
Consider the tree containing the vertex $b$ when Algorithm~\ref{alg:treer} outputs $\{a,b\}$.
The tree has at most $2r$ leaves of $T$ from the stopping condition of the algorithm.
We will show that the tree has at least $r+1$ leaves of $T$.
If $a=a^*$, the tree with root $a$ is the entire tree $T$, and since $a$ has at most three children, the tree with root $b$ contains at least $\lceil n/3\rceil\ge \lceil (3r+1)/3\rceil= r+1$ leaves of $T$.
If $a \ne a^*$, the tree with root $a$ contains at least $2r+1$ leaves of $T$ from the stopping condition.
Since $a$ has at most two children, the tree with root $b$ contains at least $\lceil (2r+1)/2\rceil=r+1$ leaves of $T$.

We assume $n=3r$. Then, the above proof works except in the case where three subtrees incident to $a^*$ each have exactly $r$ leaves.
In this case, by removing the edge $\{a^*, b\}$, the other tree rooted at $a^*$ rather than $b$ has $2r$ leaves.
\end{proof}
\begin{algorithm}[t!]
\caption{Algorithm to find an edge that satisfies the condition in Lemma~\ref{lem:part}.}\label{alg:treer}
\Input{$r\colon$ A positive integer. $T\colon$ A subcubic tree with $n\ge 3r$ leaves. }
\Output{An edge $e$ in $T$ that satisfies the condition in Lemma~\ref{lem:part}.}
\Procedure{\emph{\Call{FindEdge}{$T$, $r$}}}{
\Procedure{\emph{\Call{aux}{$T, a$}}}{
$(T', b)\leftarrow$ The rooted tree obtained by removing $a$ from $T$ that has the largest number of leaves where the root $b$ is a child of $a$ in $T$\;
\If{the number of leaves in $T'\le 2r$}{
  \Return{$\{a,b\}$}\;
}
\Return\Call{aux}{$T',b$}\;
}
$a^*\leftarrow $ Arbitrary non-leaf vertex of $T$\;
\Return\Call{aux}{$T,a^*$}\;
}
\end{algorithm}

Lemma~\ref{lem:part} implies an upper bound on the size of the minimum dependent sets.
\begin{corollary}\label{cor:2r1}
For any graph $G$ with $n$ vertices and rank-width $r\le n/3$, there exists $S\subseteq V(G)$ such that $|S|\in\{r+1,\dotsc,2r\}$ and $\cutrank_G(S)\le r$.
\end{corollary}
Hence, if $n\ge 3r$, there exists a vertex $v\in V(G)$ such that $G$ can be generated from $G-v$ with EC-complexity at most $2r-1$.
This argument implies that for any fixed $r$, the EC-complexity of graphs with $n$ vertices and rank-width $r$ is at most $(2r-1)n + O(1)$.
In the following, we will improve the constant factor $2r-1$.

First, we show the size of a dependent set $S$ can be reduced while maintaining $\cutrank_G(S) < |S|$.
\begin{lemma}\label{lem:shrink}
Let $S\subseteq V(G)$ and $k=\cutrank_G(S)$.
Assume $k < |S|$.
Then, any $T\subseteq S$ of size at least $\left\lceil\frac{k+|S|+1}2\right\rceil$ satisfies $\cutrank_G(T) < |T|$.
\end{lemma}
\begin{proof}
If $|S| - k \in \{1, 2\}$, $\left\lceil\frac{k+|S|+1}2\right\rceil = |S|$.
Hence, $T=S$ and the condition is satisfied.
Assume $|S|-k> 2$.
For any $v\in S$, let $S'=S\setminus\{v\}$.
$A_{S', V\setminus S'}$ is obtained from $A_{S, V\setminus S}$ by deleting $v$-th row and adding $v$-th column.
Hence, $\mathrm{rank}_{\mathbb{F}_2}(A_{S', V\setminus S'})\le \mathrm{rank}_{\mathbb{F}_2}(A_{S, V\setminus S})+1$.
In the process $S\to S'$, the size of the subset decreases by one while the cutrank of the subset increases by at most one.
Then,
\begin{align*}
0<
|S|-\cutrank_G(S)-2
\le
|S'|-\cutrank_G(S')
\le
|S|-\cutrank_G(S).
\end{align*}
We can repeatedly apply this process until $|S|-\cutrank_G(S) > 2$.
\end{proof}

Corollary~\ref{cor:2r1} and Lemma~\ref{lem:shrink} imply the following improved upper bound on the size of the minimum dependent sets.
\begin{corollary}\label{cor:23r}
For any graph $G$ with $n$ vertices and rank-width $r\le n/3$, there exists $S\subseteq V(G)$ such that $|S|\le\left\lceil\frac{3r+1}2\right\rceil$ and $\cutrank_G(S)< |S|$.
\end{corollary}
Corollary~\ref{cor:23r} implies that for any fixed $r$, the EC-complexity of a graph with $n$ vertices and rank-width $r$ is at most $\left\lceil\frac{3r-1}2\right\rceil n + O(1)$.

We can further improve the constant factor.
Let
\begin{align*}
C(n,r)&:= \max_{G\colon\text{$n$-vertex graph with rank-width at most $r$}}\NCO(G).
\end{align*}

\begin{lemma}\label{lem:54r}
For odd $r$ and any $n\ge 3r$, there exists $k\in\{1,2,\dotsc,r\}$ such that
\begin{align*}
C(n, r) &\le \frac{5r^2-1}{4r} k + C(n-k, r).
\end{align*}
For even $r$ and any $n\ge 3r$, there exists $k\in\{1,2,\dotsc,r\}$ such that
\begin{align*}
C(n, r) &\le \frac{5r}{4} k + C(n-k, r).
\end{align*}
\end{lemma}
\begin{proof}
Let $G$ be a graph with $n$ vertices and rank-width at most $r$ that satisfies $\NCO(G) = C(n, r)$.
There exists a set $S\subseteq V(G)$ with $|S|\in\{r+1,\dotsc,2r\}$ and $\cutrank_G(S)\le r$.
There also exists $v_1\in S$ such that $G$ can be generated from $G'=G[V(G)\setminus\{v_1\}]$ with EC-complexity at most $\left\lceil\frac{|S|+r-1}2\right\rceil$.
If $|S|\ge r+2$, $G'$ has a vertex set $S'=S\setminus\{v_1\}$ with $\cutrank_{G'}(S') \le \cutrank_{G}(S) \le r < |S'|$.
Hence, there exists $v_2\in S'$ such that $G'$ can be generated from $G'[V(G')\setminus\{v_2\}]$ with EC-complexity at most $\left\lceil\frac{|S|+r-2}2\right\rceil$.
By repeating this argument, there exist vertices $v_1,\dotsc,v_{|S|-r}$ such that $G$ can be generated from $G[V(G)\setminus\{v_1,\dotsc,v_{|S|-r}\}]$ with EC-complexity per vertex at most
\begin{align}
\frac1{|S|-r}\sum_{k=1}^{|S|-r} \left\lceil\frac{|S|+r-k}2\right\rceil
&= \frac1{|S|-r}\sum_{k=r+1}^{|S|} \left\lceil\frac{k+r-1}2\right\rceil.
\label{eq:pv}
\end{align}
Since~\eqref{eq:pv} is an average of $|S|-r$ integers, this is maximized when $|S|=2r$.
In this case, if $r$ is odd, \eqref{eq:pv} is
\begin{align*}
\frac1{r}\sum_{k=r+1}^{2r} \left\lceil\frac{k+r-1}2\right\rceil
&= \frac1r\left(r + (r+1) + (r+1) + (r+2) + (r+2) + \dotsb + \frac{3r-1}2 + \frac{3r-1}2\right)\\
&=\frac{5r^2-1}{4r},
\end{align*}
and if $r$ is even,
\begin{align*}
\frac1{r}\sum_{k=r+1}^{2r} \left\lceil\frac{k+r-1}2\right\rceil
&= \frac1r\left(r + (r+1) + (r+1) + (r+2) + (r+2) + \dotsb + \frac{3r-2}2 + \frac{3r-2}2 + \frac{3r}2\right)\\
&=\frac{5r}{4}.
\end{align*}
Hence,
\begin{align*}
C(n,r) = \NCO(G) &\le 
\begin{cases}
\frac{5r^2-1}{4r} (|S|-r) + C(n-|S|+r, r)&\text{if $r$ is odd}\\
\frac{5r}{4} (|S|-r) + C(n-|S|+r, r)&\text{otherwise.}
\end{cases}
\end{align*}
\end{proof}

Lemma~\ref{lem:54r} implies that for any fixed $r$,
\begin{align*}
C(n,r)&\le
\begin{cases}
\frac{5r^2-1}{4r}n + O(1)&\text{if $r$ is odd}\\
\frac{5r}{4}n + O(1)&\text{otherwise}
\end{cases}
\end{align*}
and hence Theorem~\ref{thm:simpu} is proved.
The constant terms $O(1)$ with respect to $n$ are calculated in Appendix~\ref{apx:ind}, which
completes the proof of Theorem~\ref{thm:upper}.

\section{Interval graphs and circle graphs}\label{sec:special}

\subsection{Definitions of graph classes}
\begin{definition}[Interval graph and circle graph]
Let $V= \{v_1,v_2,\dotsc,v_n\}$.
A \emph{double occurrence word} over the alphabet $V$ is a word in which every letter in $V$ appears exactly twice.
The \emph{interval graph} $G=(V,\,E)$ for a double occurrence word $m$ is defined as
\begin{align*}
E &= \left\{\{a, b\}\in\binom{V}{2}\,\middle|\, \text{$a$ and $b$ appear in $m$ in the order $abba$ or $abab$}\right\}.
\end{align*}
The \emph{circle graph} $G=(V,\,E)$ for a double occurrence word $m$ is defined as
\begin{align*}
E &= \left\{\{a, b\}\in\binom{V}{2}\,\middle|\, \text{$a$ and $b$ appear in $m$ in the order $abab$}\right\}.
\end{align*}
\end{definition}

It is known that the classes of interval graphs and circle graphs have unbounded rank-width~\cite{doi:10.1142/S0129054100000260,OUM2006514}.

A circle graph defined by a double occurrence word $m$ is an intersection graph of a \emph{chord diagram} determined by $m$.
Figure~\ref{fig:circle} shows an example of a double occurrence word and a corresponding chord diagram.
A circle graph for $m$ has an edge $\{v,w\}$ if and only if the chords for $v$ and $w$ intersect.
Note that the class of circle graphs is closed under vertex-minors.
The class of circle graphs, in the context of vertex-minors and rank-width, is analogous to the class of planar graphs in the context of graph-minors and branch-width.

\begin{theorem}[The grid minor theorem for vertex-minors~\cite{GEELEN202393}]
For any circle graph $H$, there exists an integer $r$ such that every graph with rank-width at least $r$ contains a vertex-minor isomorphic to $H$.
\end{theorem}
Hence, the class of circle graphs plays a key role within vertex-minor theory~\cite{KIM202454}.

\subsection{Interval graphs}
\begin{algorithm}[t!]
\caption{Interval graph preparation.}\label{alg:intervalgraph}
\Input{A double occurrence word $m$ over $\{v_1,v_2,\dotsc,v_n\}$.}
\Output{An interval graph defined by $m$ generated by the operations in Definition~\ref{def:compG}.}
\Procedure{\emph{\Call{GenIntervalGraph}{$m$}}}{
    Prepare a vertex $z$\;
    Prepare vertices $\{u_i\}_{1,\dotsc,n}$\;
    \For{$t=1,\dotsc, 2n-2$}{
        \If{the $t$-th element of $m$ is the first appearance in $m$}{
            Local complementation on $z$\;
            Edge complementation between $u_{p_1^{-1}(t)}$ and $z$ (elementary edge-complementation)\;
            Local complementation on $z$\;
        }
        \Else{
            Edge complementation between $u_{p_2^{-1}(t)}$ and $z$ (elementary edge-complementation)\;
        }
    }
    Delete $z$\;
    \Return{$\{u_1,\dotsc, u_n\}$}\;
}
\end{algorithm}

For a double occurrence word $m$ and $a\in \{1,\dotsc,n\}$, let $p_1(a)$ and $p_2(a)$ denote the positions of the first and second appearances of $v_a$ in $m$.
For example, $V = \{v_1,v_2,v_3\}$, and $m=v_3 v_2 v_2 v_1 v_3 v_1$,
\begin{align*}
p_1(1) &= 4,&
p_2(1) &= 6,&
p_1(2) &= 2,&
p_2(2) &= 3,&
p_1(3) &= 1,&
p_2(3) &= 5.
\end{align*}

\begin{theorem}
For any interval graph $G$ with $n$ vertices, Algorithm~\ref{alg:intervalgraph} generates $G$ with EC-complexity $2n-2$.
\end{theorem}
\begin{proof}
The EC-complexity of Algorithm~\ref{alg:intervalgraph} is obviously $2n-2$.
The correctness of Algorithm~\ref{alg:intervalgraph} is obvious if $t$ runs from $1$ to $2n$.
We can omit the case $t=2n$ since the $(2n)$-th element of $m$ is the second appearance in $m$.
We can also omit the case $t=2n-1$ since if the $(2n-1)$-th element of $m$ is the second appearance in $m$,
we can simply ignore this, and if the $(2n-1)$-th element of $m$ is the first appearance, then the $(2n-1)$-th and $(2n)$-th elements are equal and represent an isolated vertex in $G$.
\end{proof}
Although Algorithm~\ref{alg:intervalgraph} uses the single working vertex $z$, we can use $u_a$ in place of $z$ where $a=\arg\max p_1(a)$ without any overhead.
Hence, the EC-complexity $2n-2$ is achievable without working vertices.

\subsection{Circle graphs}
\begin{figure}[t]
\centering
\begin{tikzpicture}[scale=0.20]
\def\Radius{11}
\draw[very thick] (0, 0) circle (\Radius);
\node[draw,circle,fill, inner sep=2pt, label={[label distance=2]  0:{\LARGE $a$}}] (a0) at (  0:\Radius) {};
\node[draw,circle,fill, inner sep=2pt, label={[label distance=2] 30:{\LARGE $b$}}] (b0) at ( 30:\Radius) {};
\node[draw,circle,fill, inner sep=2pt, label={[label distance=2] 60:{\LARGE $c$}}] (c0) at ( 60:\Radius) {};
\node[draw,circle,fill, inner sep=2pt, label={[label distance=2] 90:{\LARGE $a$}}] (a1) at ( 90:\Radius) {};
\node[draw,circle,fill, inner sep=2pt, label={[label distance=2]120:{\LARGE $d$}}] (d0) at (120:\Radius) {};
\node[draw,circle,fill, inner sep=2pt, label={[label distance=2]150:{\LARGE $e$}}] (e0) at (150:\Radius) {};
\node[draw,circle,fill, inner sep=2pt, label={[label distance=2]180:{\LARGE $f$}}] (f0) at (180:\Radius) {};
\node[draw,circle,fill, inner sep=2pt, label={[label distance=2]210:{\LARGE $d$}}] (d1) at (210:\Radius) {};
\node[draw,circle,fill, inner sep=2pt, label={[label distance=2]240:{\LARGE $b$}}] (b1) at (240:\Radius) {};
\node[draw,circle,fill, inner sep=2pt, label={[label distance=2]270:{\LARGE $e$}}] (e1) at (270:\Radius) {};
\node[draw,circle,fill, inner sep=2pt, label={[label distance=2]300:{\LARGE $c$}}] (c1) at (300:\Radius) {};
\node[draw,circle,fill, inner sep=2pt, label={[label distance=2]330:{\LARGE $f$}}] (f1) at (330:\Radius) {};
\draw[-,very thick] (a0) -- (a1);
\draw[-,very thick] (b0) -- (b1);
\draw[-,very thick] (c0) -- (c1);
\draw[-,very thick] (d0) -- (d1);
\draw[-,very thick] (e0) -- (e1);
\draw[-,very thick] (f0) -- (f1);
\end{tikzpicture}
\hfill
\begin{tikzpicture}[scale=0.20]
\node[draw,circle,minimum size=25,very thick] (a) at (10,20) {\Large $a$};
\node[draw,circle,minimum size=25,very thick] (b) at (10,10) {\Large $b$};
\node[draw,circle,minimum size=25,very thick] (c) at ( 0,20) {\Large $c$};
\node[draw,circle,minimum size=25,very thick] (d) at ( 0, 0) {\Large $d$};
\node[draw,circle,minimum size=25,very thick] (e) at (10, 0) {\Large $e$};
\node[draw,circle,minimum size=25,very thick] (f) at ( 0,10) {\Large $f$};
\draw[-,very thick] (a) -- (c) -- (f) -- (d) -- (e) -- (b) -- (a);
\draw[-,very thick] (c) -- (b) -- (f) -- (e);
\end{tikzpicture}
\hfill
\begin{tikzpicture}[scale=0.20]
\node[draw,circle,minimum size=25,very thick] (a) at (30:4) {\Large $a$};
\node[draw,circle,minimum size=25,very thick] (b) at (150:4) {\Large $b$};
\node[draw,circle,minimum size=25,very thick] (c) at (270:4) {\Large $c$};
\node[draw,circle,minimum size=25,very thick] (d) at (90:14) {\Large $d$};
\node[draw,circle,minimum size=25,very thick] (e) at (210:14) {\Large $e$};
\node[draw,circle,minimum size=25,very thick] (f) at (330:14) {\Large $f$};

\draw[-,very thick] (a) -- (b) -- (c) -- (a);
\draw[-,very thick] (d) -- (e) -- (f) -- (d);
\draw[-,very thick] (d) -- (b) -- (e) -- (c) -- (f) -- (a) -- (d);
\end{tikzpicture}
\caption{Left: The chord diagram of a double occurrence word \texttt{acbafcebdfed}. Middle: The circle graph. Right: The tour graph.}
\label{fig:circle}
\end{figure}

Recently, Davies and Jena proved that for any circle graph $G$ with $n$ vertices and the chromatic number $k$, $\NCZ(\ket{G}) \le 2n\lceil \log_2 k\rceil$~\cite{davies2025preparing, jena2024graph}.
Their algorithm is a simple divide-and-conquer algorithm.
In this section, we show an alternative algorithm based on the algorithmic framework in Algorithm~\ref{alg:gengen}.

For a double occurrence word $m$ over $V$, we define a \emph{tour graph} that is a connected
(not necessarily simple) 4-regular graph that has the Eulerian cycle visiting the vertices in the order in $m$.
Figure~\ref{fig:circle} shows an example of a chord diagram, a circle graph and a tour graph corresponding to a double occurrence word.
A local complementation on a vertex $v$ in a circle graph $G$ corresponds to a reversion of the subsequence in $m$ between two occurrences of $v$.
The corresponding transformation of an Eulerian cycle in a tour graph is called the \emph{$\kappa$-transformation}~\cite{MR0248043}.
Kotzig proved that the $\kappa$-transformation is strong enough to transform an arbitrary Eulerian cycle of a connected 4-regular graph into another arbitrary Eulerian cycle.
\begin{lemma}[\cite{MR0248043}]\label{lem:kotzig}
Let $G$ be a connected 4-regular graph.
Let $m$ and $m'$ be arbitrary Eulerian cycles of $G$.
Then, $m$ can be transformed to $m'$ by a sequence of $\kappa$-transformations.
\end{lemma}
Lemma~\ref{lem:kotzig} implies that the tour graph for $G$ without a specific choice of Eulerian cycle represents the class of circle graphs local-complement equivalent to $G$.

\begin{theorem}\label{thm:circle}
For any circle graph $G$ with $n$ vertices, $\NCO(G)\le 2\left\lfloor\log_3 (n+1)\right\rfloor(n-1)\le 1.262 \cdot (n-1)\log_2 (n+1)$.
There exists an algorithm within the algorithmic framework in Algorithm~\ref{alg:gengen} satisfying the upper bound.
\end{theorem}
\begin{proof}
We will show that for any circle graph $G$ with $n$ vertices, there exists a circle graph $G'$ that is local-complement equivalent to $G$ and includes a vertex $v$ of degree at most $2\left\lfloor\log_3 (n+1)\right\rfloor$.
This implies $\NCO(G)\le 2\left\lfloor\log_3(n+1)\right\rfloor(n-1)$.
The algorithm falls within the algorithmic framework in Algorithm~\ref{alg:gengen} since $S=\{v\}\cup N_{G'}(v)$ has a cutrank at most $|S|-1$.

Let $G=(V,\,E)$ be a circle graph and $T=(V,\,F)$ be a tour graph of $G$.
Since $T$ is a 4-regular graph, $T$ has a cycle.
Let $C$ be a smallest cycle in $T$ and $v\in V(C)$ be an arbitrary vertex in $C$.
Here, a self-loop and double edges are regarded as cycles of length one and two, respectively.
Let $T'$ be the connected component in $H = (V, F-E(C))$ including $v$.
Let $m'$ be an arbitrary Eulerian cycle on $T'$ starting and ending at $v$.
There exists an Eulerian cycle on $T''=(V\setminus (V(T')\setminus V(C)),\,F-E(T'))$ since $T''$ is connected and all vertices have even degrees.
Here, $T''$ is connected since every connected component of $H$ contains at least one vertex of $C$.
Let $m''$ be an arbitrary Eulerian cycle on $T''$ starting and ending at $v$.
Let $m$ be the Eulerian cycle on $T$ obtained by the concatenation of $m'$ and $m''$.
Let $G'$ be the circle graph corresponding to the Eulerian cycle $m$ on $T$.
From Lemma~\ref{lem:kotzig}, $G'$ is local-complement equivalent to $G$.
In the Eulerian cycle for $G'$, a vertex $w\in V\setminus\{v\}$ alternates with $v$ if and only if $w$ is included in both $m'$ and $m''$.
This implies that the neighborhood of $v$ on $G'$ is a subset of $V(C)\setminus\{v\}$.
Hence, the degree of $v$ is at most $|V(C)|-1$ on $G'$.

From the Moore bound~\cite{Biggs_1974}, any 4-regular graph with $n$ vertices and girth at least $g$ must satisfy
\begin{align*}
n\ge \begin{cases}
 2\cdot 3^{\frac{g-1}2}-1&
\text{if $g$ is odd}\\
3^{\frac{g}2}-1&
\text{if $g$ is even.}
\end{cases}
\end{align*}
Hence, any graph with $n$ vertices and girth $g$ satisfies $g\le 2\left\lfloor\log_3 (n+1)\right\rfloor+1$.
This implies that the degree of $v$ is at most $2\left\lfloor\log_3(n+1)\right\rfloor$ on $G'$.
Hence, the EC-complexity of $G$ is at most $2\left\lfloor\log_3(n+1)\right\rfloor (n-1)$.
\end{proof}

Davies and Jena proved that for any circle graph $G$ that is local-complement equivalent to a bipartite graph, the EC-complexity of $G$ is at most $2n$ using a single auxiliary vertex and $3(n-1)$ without using an auxiliary vertex~\cite{davies2025preparing, jena2024graph}.
We will prove that the algorithmic framework in Algorithm~\ref{alg:gengen} gives the EC-complexity $2n$ without using an auxiliary vertex.
Every connected 4-regular planar graph has an A-trail, which is an Eulerian cycle without straight-ahead transitions~\cite{MR0248043}.
Furthermore, de Fraysseix and Ossona de Mendez proved that an A-trail always represents a bipartite circle graph, and vice versa~\cite{de_fraysseix_characterization_1999}.
\begin{lemma}[\cite{MR0248043,de_fraysseix_characterization_1999}]\label{lem:planar}
A circle graph $G$ is local-complement equivalent to a bipartite graph if and only if $G$ has a planar tour graph.
\end{lemma}

The girth of a 4-regular planar graph is upper bounded by three, which is much better than the Moore bound for general 4-regular graphs $O(\log n)$.

\begin{lemma}\label{lem:bcircle}
For any circle graph $G$ with $n$ vertices that is local-complement equivalent to a bipartite graph, $\NCO(G)\le 2(n-1)$.
There exists an algorithm within the algorithmic framework in Algorithm~\ref{alg:gengen} satisfying the upper bound.
\end{lemma}
\begin{proof}
The proof is almost the same as the proof of Theorem~\ref{thm:circle}.
Let $G$ be a circle graph with $n$ vertices that is local-complement equivalent to a bipartite graph.
From Lemma~\ref{lem:planar}, $G$ has a planar tour graph $T$.
Let $C$ be a smallest cycle of $T$.

We will show that the size of $C$ is at most three.
From Euler's formula, a planar graph with $v$ vertices, $e$ edges and $f$ faces satisfies
\begin{align*}
v-e+f=2.
\end{align*}
For 4-regular graph, $4v = 2e$.
Let $g$ be the girth of the planar graph. Then, $fg\le 2e$.
It implies
\begin{align*}
g&\le\frac{2e}{f}
=\frac{4v}{2+e-v}
=\frac{4v}{2+v}< 4.
\end{align*}
Hence, the girth of a 4-regular planar graph is at most three.

This implies that there exists a graph $G'$ that is local-complement equivalent to $G$ and includes a vertex of degree at most two.
When we remove $v$ in $T$, we can choose one of the two transitions of the Eulerian cycle at $v$ that correspond to $G'$ and $\tau_v(G')$.
The tour graphs of $G'-v$ and $\tau_v(G')-v$ are described in Figure~\ref{fig:tourp}.
Since the left resulting graph is planar, at least one of $G'-v$ and $\tau_v(G')-v$ has a planar tour graph.
It implies $\NCO(G)\le 2(n-1)$.
\end{proof}

\begin{figure}[t]
  \centering
  \begin{tikzpicture}[inner sep = 0, minimum size = 17pt, very thick, scale=1]
    \node[circle,draw] (v) at (4,0) {$v$};
    \node[circle,draw] (u) at (5,1) {$u$};
    \node[circle,draw] (w) at (5,-1) {$w$};
    \draw[->] (3,1) -> (v);
    \draw[->] (v) -> (3,-1);
    \draw[->] (v) -> (u);
    \draw[->] (u) -> (w);
    \draw[->] (w) -> (v);
    \draw (u) -- ++(45:0.6);
    \draw (u) -- ++(75:0.6);
    \draw (w) -- ++(-45:0.6);
    \draw (w) -- ++(-75:0.6);
    \draw[->] (5.8,0) -> (6.6,0);
    \node[circle,draw] (uu) at (9,1) {$u$};
    \node[circle,draw] (ww) at (9,-1) {$w$};

    \draw[->] (7,1) .. controls (8,-0.15) .. (uu);
    \draw[->] (ww) .. controls (8,0.15) .. (7,-1);
    \draw[->] (uu) -> (ww);

    \draw (uu) -- ++(45:0.6);
    \draw (uu) -- ++(75:0.6);
    \draw (ww) -- ++(-45:0.6);
    \draw (ww) -- ++(-75:0.6);
  \end{tikzpicture}
  \hfill
  \begin{tikzpicture}[inner sep = 0, minimum size = 17pt, very thick, scale=1]
    \node[circle,draw] (v) at (4,0) {$v$};
    \node[circle,draw] (u) at (5,1) {$u$};
    \node[circle,draw] (w) at (5,-1) {$w$};
    \draw[->] (3,1) -> (v);
    \draw[->] (v) -> (3,-1);
    \draw[->] (v) -> (w);
    \draw[->] (w) -> (u);
    \draw[->] (u) -> (v);
    \draw (u) -- ++(45:0.6);
    \draw (u) -- ++(75:0.6);
    \draw (w) -- ++(-45:0.6);
    \draw (w) -- ++(-75:0.6);
    \draw[->] (5.8,0) -> (6.6,0);
    \node[circle,draw] (uu) at (9,1) {$u$};
    \node[circle,draw] (ww) at (9,-1) {$w$};

    \draw[->] (7,1) -> (ww);
    \draw[white, line width=5pt] (uu) -> (7, -1);
    \draw[->] (uu) -> (7,-1);
    \draw[->] (ww) -> (uu);

    \draw (uu) -- ++(45:0.6);
    \draw (uu) -- ++(75:0.6);
    \draw (ww) -- ++(-45:0.6);
    \draw (ww) -- ++(-75:0.6);
  \end{tikzpicture}
  \caption{Two possible tour-graph reductions obtained by deleting $v$ of degree two in the circle graph. In the left Eulerian cycle, the vertices are visited in the order $v,\,u,\,w,\,v$. In the right Eulerian cycle, the vertices are visited in the order $v,\,w,\,u,\,v$. The two resulting graphs are tour graphs of $G'-v$ and $\tau_v(G')-v$.}
  \label{fig:tourp}
\end{figure}

The same proof for Lemma~\ref{lem:bcircle} does work if $T$ can be embedded in a surface of Euler genus $1$.
\begin{lemma}
For any circle graph $G$ with $n$ vertices whose tour graph can be embedded in a surface of Euler genus one, $\NCO(G)\le 2(n-1)$.
There exists an algorithm within the algorithmic framework in Algorithm~\ref{alg:gengen} satisfying the upper bound.
\end{lemma}
\begin{proof}
The proof is almost the same as Lemma~\ref{lem:bcircle}.
We use generalized Euler's formula $v-e+f\ge 1$ instead.
\end{proof}

For general Euler genus, we obtain a similar upper bound.
\begin{lemma}
Let $\mathcal{C}$ be a class of circle graphs whose tour graph can be embedded in a surface of Euler genus $\epsilon$.
The EC-complexity of graphs in $\mathcal{C}$ is at most $3 n + O(\epsilon\log\epsilon)$.
There exists an algorithm within the algorithmic framework in Algorithm~\ref{alg:gengen} satisfying the upper bound.
\end{lemma}
\begin{proof}
The proof is almost the same as the proof of Lemma~\ref{lem:bcircle}.
By generalized Euler's formula,
\begin{align*}
2-\epsilon\le v-e+f.
\end{align*}
For $v\ge \max\{0,\,\epsilon-1,\,5\epsilon-9\}=\max\{0,\,5\epsilon-9\}$, $fg\le 2e$ implies
\begin{align*}
g&\le\frac{2e}{f}
\le\frac{4v}{2-\epsilon+e-v}
=\frac{4v}{2-\epsilon+v}
=5 - \frac{10-5\epsilon+v}{2-\epsilon+v}
< 5.
\end{align*}
Hence, the girth of $T$ is at most four if $n\ge \max\{0,\,5\epsilon-9\}$.
From Theorem~\ref{thm:circle}, EC-complexity of a circle graph with $5\epsilon-10$ vertices is $O(\epsilon\log\epsilon)$.
This implies the upper bound $3n+O(\epsilon\log\epsilon)$.
\end{proof}

Lemma~\ref{lem:planar} implies the equivalence of the planarity of tour graph and the bipartiteness of circle graph up to the local complementation.
It would be interesting to see how the Euler genus of tour graph is reflected to properties of corresponding local-complement equivalence class of circle graphs.

Algorithm~\ref{alg:intervalgraph} for interval graphs is the unique graph generation algorithm in this paper that is not included within the algorithmic framework in Algorithm~\ref{alg:gengen}.
It is an open question whether Algorithm~\ref{alg:gengen} gives algorithms with EC-complexity $O(n)$ for interval graphs.

\section{Conclusion}

In this work, we studied the complexity of graph-state preparation in a general model of quantum algorithms that allows measurements in the computational basis, single-qubit Clifford operations, and two-qubit Clifford operations, and introduced the CZ-complexity, defined as the minimum number of two-qubit Clifford operations required to generate a target graph state.
We first showed that this complexity measure is well justified, since every optimal graph-state preparation algorithm can be assumed to use only \CZ operations as its two-qubit gates.
We then established a combinatorial characterization of graph-state transformations and proved that the CZ-complexity of a graph state coincides with the edge-complementation-complexity of the corresponding graph.
More precisely, a graph state $\ket{G}$ can be generated from $\ket{H}$ with CZ-complexity at most $t$ if and only if $G$ can be obtained from $H$ by vertex deletions, local complementations, and at most $t$ elementary edge-complementations.
This yields a purely graph-theoretic description of graph-state preparation.

Based on this characterization, we related graph-state preparation complexity to rank-width.
For an $n$-vertex graph of rank-width at most $r$, we obtained an upper bound of $O(rn)$ on the CZ-complexity, while for connected graphs of rank-width at least $r$ we proved the lower bound $n+r-2$.
We also showed that these bounds are close to optimal in two complementary senses.
There exist graphs of rank-width at most $r$ whose CZ-complexity is $\Omega(rn/\log n)$, and there exist connected graphs for which the additive lower bound beyond $n-1$ is tight up to a constant factor.
In particular, for connected graphs, rank-width one exactly characterizes the graph states with CZ-complexity $n-1$.

We also considered special graph classes with unbounded rank-width.
For interval graphs and circle graphs, we presented preparation algorithms with CZ-complexities $O(n)$ and $O(n\log n)$, respectively.
These examples show that although rank-width provides a useful general framework for bounding graph-state preparation complexity, finer combinatorial descriptions can yield substantially better bounds for specific graph families.


Measurement-based quantum computation (MBQC) is one of the main applications of graph states.
It is known that MBQC on graph states of logarithmically bounded rank-width can be simulated efficiently on a classical computer~\cite{PhysRevA.75.012337}.
Efficient classical simulation is also known for MBQC on circle graph states~\cite{PhysRevA.76.022304,harrison2025fermion,hahn2026structure}.
From the perspective of MBQC, a natural goal is therefore to identify classes of graph states that admit efficient preparation but for which efficient classical simulation of MBQC is not known.
In view of the results of this paper, natural candidate classes include interval graphs and graph classes of moderate rank-width, for example rank-width $r=\Theta(\sqrt{n})$.

Several other problems remain open.
It would be worthwhile to sharpen the constants in the rank-width-based bounds and to identify other local-complementation-invariant graph parameters that yield meaningful lower bounds on CZ-complexity.
On the algorithmic side, it remains open whether the general framework developed in this paper can also achieve $O(n)$ EC-complexity for interval graphs.
For circle graphs, it would be natural to understand how the Euler genus of a tour graph is reflected in the structure of the corresponding local-complementation equivalence class.
All preparation algorithms developed in this paper avoid the use of working qubits. It remains unclear whether allowing working qubits can reduce the number of \CZ operations. Clarifying this question, and more generally understanding the trade-off between the number of working qubits and the number of \CZ operations, would be important for practical implementations.
Another natural direction is to enlarge the class of single-qubit operations that do not contribute to the two-qubit count and allow arbitrary single-qubit unitaries.
In this direction, Claudet and Perdrix recently gave a graphical characterization of local unitary equivalence in terms of generalized local complementation~\cite{claudet_et_al:LIPIcs.STACS.2025.27}.
Developing an analogous graphical characterization for transformations generated by an arbitrary number of local unitary operations together with at most $t$ \CZ operations would substantially broaden the scope of the present framework.
We hope that the combinatorial viewpoint developed in this paper will be useful for further studies on graph-state preparation, stabilizer-state transformations, and resource states for MBQC.

\section*{Acknowledgments}
The work of RM was supported by JST FOREST Program Grant Number JPMJFR216V and JSPS KAKENHI Grant Numbers JP20H04138, JP20H05966 and JP22H00522.
We thank Simon Martiel and Tristan Cam for identifying errors in the algorithms for generating permutation graphs and circle graphs in the earlier version of the paper.
We also thank James Davies and Andrew Jena for identifying the same errors independently.

\paragraph{Author contributions.}

The notion of CZ-complexity and the problem of minimizing it were proposed by Yusei Yoshimura.
Yusei Yoshimura and Ryuhei Mori proved preliminary versions of some of the results in this paper.
Later, Soh Kumabe and Ryuhei Mori strengthened these results using the cut-rank and rank-width.
Large language models were used only to assist with the writing of the manuscript,
including language editing, phrasing, and clarity improvements. They were not used
to produce mathematical results, proofs, calculations, code, or figures. All authors
approved the final manuscript and are fully responsible for its content.

\bibliographystyle{quantum}
\bibliography{biblio}

\appendices

\section{Proofs of Facts~\ref{fa:sgn} and~\ref{fa:L}}\label{apx:sgnL}
\subsection{Proof of Fact~\ref{fa:sgn}}
\fasgn*
\begin{proof}
Let $G$ be the common generating matrix for $\ket{\psi}$ and $\ket{\varphi}$.
Let $v\in\mathbb{F}_2^n$ be a vector where $v_i=1$ if and only if the sign of the $i$-th stabilizer for $\ket{\psi}$ and $\ket{\varphi}$ differs.
There exists a row vector $a\in\mathbb{F}_2^{2n}$ satisfying
\begin{align*}
G
\begin{bmatrix}
O_n&I_n\\
I_n&O_n
\end{bmatrix}
a^T
&=v
\end{align*}
since
$G
\begin{bmatrix}
O_n&I_n\\
I_n&O_n
\end{bmatrix}\in\mathbb{F}_2^{n\times 2n}$ has linearly independent rows.
Then, $S_i \sigma_a\ket{\varphi} = (-1)^{v_i} \sigma_a\ket{\varphi}$ where $S_i$ is the $i$-th stabilizer of $\ket{\varphi}$.
Hence, $\sigma_a \ket{\varphi}$ has the same stabilizer as $\ket{\psi}$.
\end{proof}

\subsection{Proof of Fact~\ref{fa:L}}
\faL*
\begin{proof}
First, we prove that $L_C = I_{\mathbb{F}_2^{2n}}$ implies $C$ is in $\mathcal{G}_n$ up to a constant factor.
Assume $L_C = I_{\mathbb{F}_2^{2n}}$.
Then, for any $h\in\mathcal{G}_n$, there exists $\alpha_h\in\{\pm1\}$ such that $ChC^\dagger = \alpha_h h$.
For any $h_1,h_2\in\mathcal{G}_n$, $\alpha_{h_1h_2} h_1 h_2 = Ch_1h_2C^\dagger = C h_1 C^\dagger C h_2 C^\dagger= \alpha_{h_1}\alpha_{h_2} h_1 h_2$.
Hence, $\alpha_{h_1h_2} = \alpha_{h_1}\alpha_{h_2}$ for any $h_1,h_2\in\mathcal{G}_n$.
This implies that the map $h\mapsto \alpha_h$ is a group homomorphism from $\mathcal{G}_n$ to $\{\pm1\}$.

Let $a\in\mathbb{F}_2^{2n}$ be a vector satisfying
\begin{align*}
\alpha_{X_i}&=(-1)^{a_{n+i}},&\alpha_{Z_i}&=(-1)^{a_{i}}\qquad\forall i\in\{1,2,\dotsc,n\}.
\end{align*}
Then, $g=\sigma_a$ satisfies
\begin{align*}
gX_i &= \alpha_{X_i} X_i g,&
gZ_i &= \alpha_{Z_i} Z_i g\qquad\forall i\in\{1,2,\dotsc,n\}.
\end{align*}
Since every element of $\mathcal{G}_n$ is a product of the generators
$X_1,\dotsc,X_n,Z_1,\dotsc,Z_n$ up to a constant factor,
$gh=\alpha_h hg$
for any $h\in\mathcal{G}_n$ from the homomorphism.
Then, $(Cg)h (Cg)^\dagger = \alpha_h ChC^\dagger = \alpha^2_h  h = h$.
Hence, $Cg$ commutes with any matrix in $\mathcal{G}_n$, and hence commutes with any matrix in $U(2^n)$.
This means that $Cg$ is equal to the identity matrix up to a constant factor.
Hence, $C$ is in $\mathcal{G}_n$ up to a constant factor.

Next, assume $L_C=L_D$. Then, $L_{C D^{-1}} = L_{D^{-1}} L_C = L_D^{-1} L_C = I_{\mathbb{F}_2^{2n}}$.
This means that there exists $g\in\mathcal{G}_n$ such that $C D^{-1}$ is equal to $g$ up to a constant factor.
\end{proof}

\section{Sufficiency of the \CZ operators: Proof of Proposition~\ref{prop:cz0}}\label{apx:cz}
In this section, we prove Proposition~\ref{prop:cz0}.
\propcz*
We first prove that any two-qubit Clifford gate can be replaced by single-qubit gates, at most one \CZ gate and at most one SWAP gate.

\begin{lemma}\label{lem:cz}
Any two-qubit Clifford operator is represented by a quantum circuit consisting of arbitrary number of single-qubit Clifford gates, at most one \CZ gate and at most one SWAP gate.
\end{lemma}
\begin{proof}
From Fact~\ref{fa:L}, for a given two-qubit Clifford operator $\mathcal{C}$, it is sufficient to show the existence of the quantum circuit $\mathcal{D}$ satisfying $L_\mathcal{C}=L_\mathcal{D}$ and the conditions of the lemma.
Since the $S$, $H$ and \CZ gates generate the Clifford group, any two-qubit Clifford operator is represented by a Clifford circuit consisting of these gates.
When the Clifford circuit includes more than one \CZ gates, we can reduce the number of \CZ gates by the following procedure.
There are six possible single-qubit Clifford gates (up to succeeding Pauli gates) in~\eqref{eq:1clifford}.
Since $S$ commutes with \CZ, we can assume without loss of generality that single-qubit gates between two specific \CZ gates are all Hadamard gates.
In the two possible cases, the number of \CZ gates can be reduced by using the SWAP gate as follows.
\begin{align*}
\Qcircuit @C=1em @R=.7em {
& \ctrl{1} & \gate{H} & \ctrl{1} & \qw \\
& \control\qw & \qw & \control\qw & \qw \\
}
\quad&\raisebox{-1em}{=}\quad
\Qcircuit @C=1em @R=.7em {
& \gate{H} & \targ & \control\qw & \qw \\
& \qw & \ctrl{-1} & \ctrl{-1} & \qw \\
}
\quad\raisebox{-1em}{=}\quad
\Qcircuit @C=1em @R=.7em {
& \gate{H} & \gate{iY} & \qw \\
& \qw & \ctrl{-1} & \qw \\
}
\quad\raisebox{-1em}{=}\quad
\Qcircuit @C=1em @R=.7em {
& \gate{H} & \gate{Y} & \qw&\qw \\
& \qw & \ctrl{-1} & \gate{S}&\qw  \\
}\\
\quad&\raisebox{-1em}{=}\quad
\Qcircuit @C=1em @R=.7em {
&\gate{H}& \gate{S^\dagger} & \gate{H} & \ctrl{1}    & \gate{H} &\gate{S}& \qw \\
& \qw&\qw      & \qw      & \control\qw & \gate{S} &\qw& \qw \\
}\\
\Qcircuit @C=1em @R=.7em {
& \ctrl{1} & \gate{H} & \ctrl{1} & \qw \\
& \control\qw & \gate{H} & \control\qw & \qw \\
}
\quad&\raisebox{-1em}{=}\quad
\Qcircuit @C=1em @R=.7em {
& \qswap & \gate{H} & \ctrl{1} & \gate{H} & \qw \\
& \qswap\qwx & \gate{H} & \control\qw&\gate{H} & \qw \\
}
\end{align*}
The second relation is obtained from the following Clifford circuit for the SWAP gate.
\begin{equation*}
\Qcircuit @C=1em @R=2em {
&\qswap&\qw\\
&\qswap\qwx&\qw
}
\quad\raisebox{-1.0em}{=}\quad
\Qcircuit @C=1em @R=2em {
&\qw&\qswap&\gate{H}&\qw\\
&\gate{H}&\qswap\qwx&\qw&\qw
}
\quad\raisebox{-1.7em}{=}\quad
\Qcircuit @C=1em @R=2em {
&\qw&\ctrl{1}&\targ&\ctrl{1}&\gate{H}&\qw\\
&\gate{H}&\targ&\ctrl{-1}&\targ&\qw&\qw
}
\quad\raisebox{-1.7em}{=}\quad
\Qcircuit @C=1em @R=2em {
&\ctrl{1}&\gate{H}&\control\qw&\gate{H}&\ctrl{1}&\gate{H}&\qw\\
&\control\qw&\gate{H}&\ctrl{-1}&\gate{H}&\control\qw&\gate{H}&\qw
}
\end{equation*}
By repeating this procedure, the number of the \CZ gates is reduced to at most one.
\end{proof}

Then, by replacing all two-qubit Clifford gates in quantum circuits with Clifford circuits with at most one CZ gate and SWAP gate, we obtain the following theorem.

\begin{lemma}\label{lem:czc}
Let $\mathcal{C}$ be a Clifford circuit that consists of arbitrary number of single-qubit Clifford gates and $t$ two-qubit Clifford gates.
Then, there exists a Clifford circuit $\mathcal{D}$ that consists of arbitrary number of single-qubit Clifford gates and at most $t$ \CZ gates and a permutation of qubits $\mathcal{P}$ such that $\mathcal{D}\mathcal{P} = \mathcal{C}$.
\end{lemma}
\begin{proof}
From Lemma~\ref{lem:cz}, all two-qubit Clifford gates in $\mathcal{C}$ can be replaced by Clifford circuits consisting of arbitrary number of single-qubit Clifford gates, at most one \CZ gates and at most one SWAP gates.
The SWAP gate commutes with an arbitrary quantum gate $U$ by changing the qubits for which $U$ is applied.
Hence, the SWAP gate can be placed at the beginning of the quantum circuit, and obtain Lemma~\ref{lem:czc}.
\end{proof}

Now, we are ready to prove Proposition~\ref{prop:cz0}.

\begin{proof}[Proof of Proposition~\ref{prop:cz0}]
From Lemma~\ref{lem:paulim}, regardless of the measurement outcome, an adaptive local Clifford operation gives some graph state deterministically.
Hence, it is sufficient to show that there is a quantum algorithm that includes at most $\NCZ(\ket{G})$ \CZ operations and no other two-qubit Clifford operations, and produces $\ket{G}$ when all measurement outcomes are zero.
There exist a Clifford circuit $\mathcal{C}$ 
including arbitrary number of single-qubit Clifford gates and $\NCZ(\ket{G})$ two-qubit Clifford gates such that
\begin{align*}
\bra{0}^{\otimes s}\mathcal{C}\ket{0}^{\otimes(n+s)} &\propto \ket{G}
\end{align*}
From Lemma~\ref{lem:czc},
there exists a Clifford circuit $\mathcal{D}$ including some single-qubit Clifford gates and at most $\NCZ(\ket{G})$ \CZ gates such that
\begin{align*}
\bra{0}^{\otimes s}\mathcal{D}\ket{0}^{\otimes(n+s)} &\propto \ket{G}
\end{align*}
This proves Proposition~\ref{prop:cz0}.
\end{proof}

\begin{figure}
\centering
\begin{subfigure}[t]{0.24\textwidth}
\begin{align*}
\Qcircuit @C=1em @R=2em {
&\gate{A}&\qw\\
&\gate{B}&\qw
}
\end{align*}
\caption{Zero \CZ gate.}
\end{subfigure}
\begin{subfigure}[t]{0.24\textwidth}
\begin{align*}
\Qcircuit @C=1em @R=2em {
&\gate{C}&\ctrl{1}&\gate{A}&\qw\\
&\gate{D}&\control\qw&\gate{B}&\qw
}
\end{align*}
\caption{One \CZ gate.}
\label{fig:2clifford1}
\end{subfigure}
\begin{subfigure}[t]{0.24\textwidth}
\begin{align*}
\Qcircuit @C=1em @R=2em {
&\qswap&\gate{C}&\ctrl{1}&\gate{A}&\qw\\
&\qswap\qwx&\gate{D}&\control\qw&\gate{B}&\qw
}
\end{align*}
\caption{Two \CZ gates.}
\end{subfigure}
\begin{subfigure}[t]{0.24\textwidth}
\begin{align*}
\Qcircuit @C=1em @R=2em {
&\qswap&\gate{A}&\qw\\
&\qswap\qwx&\gate{B}&\qw
}
\end{align*}
\caption{Three \CZ gates.}
\end{subfigure}
\caption{Classification of two-qubit Clifford operators (up to succeeding Pauli gates): The gates A and B are chosen from the six single-qubit Clifford operators in~\eqref{eq:1clifford}. The gates C and D are chosen from $I$, $H$ and $HS$.}
\label{fig:2clifford}
\end{figure}
\begin{remark}
In terms of the number of \CZ gates when the SWAP gate is not allowed, two-qubit Clifford operators are classified as in Figure~\ref{fig:2clifford}.
The number of two-qubit Clifford operators of types (a), (b), (c) and (d) are $36=6^2$, $324=3^2\times 6^2$, $324$ and $36$, respectively.
Note that the numbers of required $\mathrm{CNOT}$ gates (equivalently the \CZ gates) for Clifford operators up to six qubits are calculated by computer search in~\cite{bravyi20226}.
\end{remark}

\section{Independence systems, additive codes and the EC-complexity}\label{apx:ind}
\subsection{Isotropic independence systems}

In this appendix, we consider upper bounds on the minimum dependent set that depend only on $n$, and prove Theorem~\ref{thm:upper}.
For any graph $G$, let
\begin{align*}
\mathcal{I}_G &:=\left\{S\subseteq V(G)\mid \cutrank_G(S)=|S|\right\}.
\end{align*}
We first regard the structure $(V(G),\,\mathcal{I}_G)$ as an independence system in matroid theory.
We then switch the understanding of the structure from matroid theory to coding theory.

\begin{definition}[Independence system]
Let $V$ be a finite set and $\mathcal{I}\subseteq 2^V$.
Then, $(V,\,\mathcal{I})$ is an \emph{independence system} if it satisfies the following conditions:
\begin{enumerate}
\item $\varnothing\in\mathcal{I}$.
\item For any $X\in\mathcal{I}$ and $Y\subseteq X$, $Y\in\mathcal{I}$.
\end{enumerate}
\end{definition}

Then, $(V(G), \mathcal{I}_G)$ is obviously an independence system.
Furthermore, $(V(G), \mathcal{I}_G)$ is an isotropic independence system.

\begin{definition}[Linear 2-extendible system and isotropic independence system]
An independence system $(V, \mathcal{I})$ is a \emph{linear 2-extendible system} on a field $\mathbb{F}$ if
there exists $m\in\mathbb{N}$ and two $m$-dimensional vectors $a^{v}\in\mathbb{F}^m$ and $b^v\in\mathbb{F}^m$ associated to each $v\in V$, and
\begin{align*}
\mathcal{I} &=\left\{S\subseteq V\mid \text{A set of $2|S|$ vectors $a^v$ and $b^v$ for $v\in S$ are linearly independent}\right\}.
\end{align*}
Especially, a linear 2-extendible system on $\mathbb{F}_2$ is a \emph{binary 2-extendible system}.

A binary 2-extendible system is an \emph{isotropic independence system} if
$m=|V|$ and for any $s,\,t\in\{1,2,\dotsc,|V|\}$
\begin{align}
\sum_{v\in V} \left(a^v_s b^v_t + b^v_s a^v_t\right) &= 0\label{eq:isotropic}
\end{align}
where $a^v_s\in\mathbb{F}_2$ denotes an $s$-th element of $a^v\in\mathbb{F}_2^{|V|}$.
\end{definition}

\begin{lemma}
For any graph $G$, $(V(G), \mathcal{I}_G)$ is an isotropic independence system.
\end{lemma}
\begin{proof}
For each $v\in V(G)$, let $a^v,\,b^v\in\mathbb{F}_2^n$ be vectors defined as
\begin{equation}
\begin{split}
\text{A $w$-th element of $a^v$ is 1} &\iff w=v\\
\text{A $w$-th element of $b^v$ is 1} &\iff \{v,w\}\in E.
\end{split}\label{eq:ab}
\end{equation}
We first confirm that $((a^v, b^v))_{v\in V(G)}$ satisfies the isotropic condition~\eqref{eq:isotropic}.
Let $A$ be the adjacency matrix of $G$. Then,
\begin{align*}
\sum_{v\in V(G)} \left(a^v_s b^v_t + b^v_s a^v_t\right) &= 
\sum_{v\in V(G)} \left(\delta_{v, s} A_{v,t} + A_{v,s} \delta_{v, t}\right)\\
&= A_{s,t} + A_{t,s} = 0
\end{align*}
Hence, $((a^v, b^v))_{v\in V(G)}$ satisfies the isotropic condition~\eqref{eq:isotropic}.

We now confirm that the binary 2-extendible system defined by $((a^v, b^v))_{v\in V(G)}$ represents $(V(G),\,\mathcal{I}_G)$.
For any $S\subseteq V(G)$,
\begin{align*}
\dim\mathrm{span}_{\mathbb{F}_2}\left(\bigcup_{v\in S}\{a_v,b_v\}\right)&= 
\mathrm{rank}_{\mathbb{F}_2}\left(\begin{bmatrix}
I_S&A_{S,S}\\
0&A_{V(G)\setminus S,S}
\end{bmatrix}\right)\\
&= |S| + \mathrm{rank}_{\mathbb{F}_2}(A_{V(G)\setminus S,S})\\
&=|S| + \cutrank_G(S).
\end{align*}
Hence, all the $2|S|$ vectors are linearly independent if and only if
\begin{align*}
|S| + \cutrank_G(S) = 2|S|&\iff
\cutrank_G(S) = |S|.
\end{align*}
This implies that the binary 2-extendible system represents $(V(G),\mathcal{I}_G)$.
\end{proof}

The isotropic system was introduced by Bouchet~\cite{BOUCHET1987231}. However, the isotropic system has not been regarded as an independence system.
Note that for a bipartite graph $G=(A\cup B, E)$, the binary 2-extendible system $(A\cup B,\mathcal{I}_G)$
is an intersection of the binary matroid $\mathcal{M}=\mathrm{Bin}(G, A, B)$ and its dual $\mathcal{M}^*=\mathrm{Bin}(G, B, A)$, defined in~\cite{OUM200579},
i.e., $S\in\mathcal{I}_G$ if and only if $S$ is independent both for $\mathcal{M}$ and $\mathcal{M}^*$.
Generally, for any stabilizer state with generating matrix $G\in\mathbb{F}_2^{n\times 2n}$, there exists corresponding isotropic independence system $(V,\, \mathcal{I})$.
Here, $V$ corresponds to the set of qubits, and $\mathcal{I}$ is an independent set defined by $((a^v, b^v))_{v\in V}$ where $a^v$ and $b^v$ are $v$-th column and $(n+v)$-th column of $G$, respectively.

\subsection{Self-dual additive code over $\mathbb{F}_4$ with respect to the Hermitian trace inner product}
Since from Lemma~\ref{lem:grow}, we are interested in small dependent sets of the isotropic independence system $(V(G),\,\mathcal{I}_G)$ rather than large independent sets, we switch our perspective from matroid theory to coding theory.
Indeed, self-dual additive codes over $\mathbb{F}_4$ with respect to the Hermitian trace inner product are essentially equivalent to isotropic systems~\cite{DANIELSEN20061351}.
Let $\mathbb{F}_4 = \{0,1,\omega,\omega^2\}$ where $\omega^2+\omega+1=0$.
The \emph{conjunction} of $x\in\mathbb{F}_4$ is defined as $\overline{x} = x^2$.
The \emph{trace map} $\mathrm{Tr}\colon\mathbb{F}_4\to\mathbb{F}_2$ is defined as $\mathrm{Tr}(x) = x + \overline{x}$.
\begin{definition}[Self-dual additive code over $\mathbb{F}_4$ with respect to the Hermitian trace inner product]
A non-empty subset $\mathcal{C}\subseteq \mathbb{F}_4^n$ is an \emph{additive code} over $\mathbb{F}_4$ if $x + y \in \mathcal{C}$ for all $x,\,y\in\mathcal{C}$.
Each element $x\in\mathcal{C}$ is called a \emph{codeword} of $\mathcal{C}$.
A \emph{support} of a codeword $x\in\mathcal{C}$ is defined as
\begin{align*}
\mathrm{support}(x)&:=\{i\in\{1,2,\dotsc,n\}\mid x_i\ne 0\}.
\end{align*}
For $x,\, y\in\mathbb{F}_4^n$, the Hermitian trace inner product is defined as
\begin{align*}
\langle x,\,y\rangle &= \sum_{i=1}^n \mathrm{Tr}(x_i\overline{y_i}).
\end{align*}
Two vectors $x,\,y\in\mathbb{F}_4^n$ with $\langle x,\,y\rangle=0$ are said to be \emph{orthogonal}.
The dual code $\mathcal{C}^\perp\subseteq\mathbb{F}_4^n$ of additive code $\mathcal{C}$ is defined as
\begin{align*}
\mathcal{C}^\perp &= \left\{y\in\mathbb{F}_4^n\mid \langle x, y\rangle = 0\quad\forall x\in\mathcal{C}\right\}.
\end{align*}
An additive code $\mathcal{C}$ over $\mathbb{F}_4$ is said to be \emph{self-dual with respect to the Hermitian trace inner product} if $\mathcal{C}^\perp=\mathcal{C}$.
Self-dual additive code over $\mathbb{F}_4$ with respect to the Hermitian trace inner product is denoted by $4^{\mathrm{H}+}$-code.

An additive code over $\mathbb{F}_4$ is represented by a \emph{generator matrix} $H\in\mathbb{F}_4^{k\times n}$ satisfying
\begin{align*}
\mathcal{C} &= \mathrm{span}_{\mathbb{F}_2}(\text{Rows of } H).
\end{align*}
When $\mathcal{C}$ is self-dual, $k=n$ and there is another representation using the generator matrix $H$
\begin{align*}
\mathcal{C} &= \left\{y\in\mathbb{F}_4^n\mid \langle x, y\rangle = 0\quad \text{ for any row $x$ of } H\right\}.
\end{align*}
\end{definition}
We introduce an isomorphism between $\mathcal{G}_1/\langle iI_2\rangle$ and $\mathbb{F}_4$ defined as
\begin{align*}
I_2&\longleftrightarrow 0,&
X&\longleftrightarrow \omega,&
Y&\longleftrightarrow \omega^2,&
Z&\longleftrightarrow 1.
\end{align*}
We can also introduce an isomorphism between $\mathcal{G}_n/\langle iI_2^{\otimes n}\rangle$ and $\mathbb{F}_4^n$ by applying the isomorphism for each entry.
It is easy to confirm that two Pauli matrices $x,\,y\in\mathcal{G}_n/\langle iI_2^{\otimes n}\rangle$ commute if and only if the Hermitian trace inner product of their $\mathbb{F}_4^n$ representations is zero.
Hence, there is one-to-one correspondence between stabilizer states and $4^{\mathrm{H}+}$-codes.
An example of the correspondence is shown in Figure~\ref{fig:codes}.

\begin{figure}
\begin{align*}
\begin{bmatrix}
X & X & I & I\\
I & X & X & I\\
I & I & X & X\\
Z & Y & Y & Z\\
\end{bmatrix}
\longleftrightarrow
\left(
\begin{bmatrix}
1\\
0\\
0\\
0\\
\end{bmatrix},
\begin{bmatrix}
0\\
0\\
0\\
1\\
\end{bmatrix}\right),
\left(
\begin{bmatrix}
1\\
1\\
0\\
1\\
\end{bmatrix},
\begin{bmatrix}
0\\
0\\
0\\
1\\
\end{bmatrix}\right),
\left(
\begin{bmatrix}
0\\
1\\
1\\
1\\
\end{bmatrix},
\begin{bmatrix}
0\\
0\\
0\\
1\\
\end{bmatrix}\right),
\left(
\begin{bmatrix}
0\\
0\\
1\\
0\\
\end{bmatrix},
\begin{bmatrix}
0\\
0\\
0\\
1\\
\end{bmatrix}\right)
\longleftrightarrow
\begin{bmatrix}
\omega & \omega & 0 & 0\\
0 & \omega & \omega & 0\\
0 & 0 & \omega & \omega\\
1 & \omega^2 & \omega^2 & 1\\
\end{bmatrix}
\end{align*}
\caption{Correspondence between the stabilizer generators, isotropic independence system and generator matrix of $4^{\mathrm{H}+}$-code.
}\label{fig:codes}
\end{figure}

For a graph state $\ket{G}$ with adjacency matrix $A\in\mathbb{F}_2^{n\times n}$, the corresponding $4^{\mathrm{H}+}$-code with generator matrix $\omega I + A$ is denoted by $\mathcal{C}_G$.

\begin{lemma}\label{lem:md}
For any graph $G$ and non-empty $S\subseteq V(G)$, $S\notin\mathcal{I}_G$ if and only if there exists a non-zero codeword of $\mathcal{C}_G$ whose support is a subset of $S$.
\end{lemma}
\begin{proof}
For any $a,b,s,t\in\mathbb{F}_2$,
\begin{align*}
\mathrm{Tr}\left((b+a\omega)\overline{(s+t\omega)}\right)
&= \mathrm{Tr}\left((b+a\omega)(s+t\omega^2)\right)\\
&= \mathrm{Tr}\left(bs + at + bt\omega^2 + as\omega\right)\\
&= sa + tb
\end{align*}
Let $a^v,b^v\in\mathbb{F}_2^n$ be vectors in~\eqref{eq:ab} for each $v\in V(G)$.
Then, $\omega a^v + b^v$ is the $v$-th column of the generator matrix $H$ of $\mathcal{C}_G$.
For any graph $G$ and non-empty $S\subseteq V(G)$,
\begin{align*}
S\notin\mathcal{I}_G&\iff\exists s, t\in\mathbb{F}_2^S,\qquad \lnot(s=t=0^S)\;\land\; \sum_{v\in S} (s_va^v + t_vb^v) = 0^{V(G)}\\
&\iff\exists s, t\in\mathbb{F}_2^S,\qquad \lnot(s=t=0^S)\;\land\; \sum_{v\in S} (s_va_w^v + t_vb_w^v) = 0\qquad\forall w\in V(G)\\
&\iff\exists s, t\in\mathbb{F}_2^S,\qquad \lnot(s=t=0^S)\;\land\; \sum_{v\in S} \mathrm{Tr}\left((b^v_w + a^v_w\omega)\overline{(s_v+t_v\omega)}\right) = 0\qquad\forall w\in V(G)\\
&\iff\exists y\in\mathbb{F}_4^S\setminus\{0^S\},\qquad \sum_{v\in S} \mathrm{Tr}\left(H_{wv}\,\overline{y_v}\right) = 0\qquad\forall w\in V(G)\\
&\iff\exists y\in\mathbb{F}_4^S\setminus\{0^S\}, \qquad \sum_{v\in S}\mathrm{Tr}(x_v\,\overline{y}_v) = 0\qquad\text{for any row $x$ of $H$}\\
&\iff\exists y\in\mathbb{F}_4^{V(G)}\setminus\{0^{V(G)}\}, \qquad \mathrm{support}(y)\subseteq S\;\land\; \langle x,y\rangle = 0\qquad\text{for any row $x$ of $H$}\\
&\iff\exists y\in\mathbb{F}_4^{V(G)}\setminus\{0^{V(G)}\}, \qquad \mathrm{support}(y)\subseteq S\;\land\; y\in\mathcal{C}_G.
\qedhere
\end{align*}
\end{proof}

Since we are interested in small dependent sets of the isotropic independence system $(V(G),\,\mathcal{I}_G)$,
we are interested in low-weight codewords of $4^{\mathrm{H}+}$-code $\mathcal{C}_G$.
Note that the minimum distance of $\mathcal{C}_G$ is one plus the minimum degree of vertices in graphs that are LC-equivalent to $G$~\cite{DANIELSEN20061351,hoyer2006resources,danielsen2009graph}.
Now, we can apply results in coding theory.

\begin{lemma}[\cite{rains2002self,self-dual2006}]\label{lem:mind}
Let $\mathcal{C}$ be a $4^{\mathrm{H}+}$-code.
Then, the minimum distance of $\mathcal{C}$ is at most $2\lfloor\frac{n}6\rfloor +2$ for $n\not\equiv 5\bmod 6$,
and $2\lfloor\frac{n}6\rfloor+3$ for $n\equiv5\bmod6$.
\end{lemma}
It is conjectured that the minimum distance is at most $2\lfloor\frac{n}6\rfloor+1$ for $n\equiv1\bmod 6$,
which was confirmed for $n\le 25$~\cite{self-dual2006}.

We now prove the upper bound of the EC-complexity.
Let
\begin{align*}
C(n)&:= \max_{G\colon\text{$n$-vertex graph}}\NCO(G).
\end{align*}

\begin{lemma}\label{lem:wcost}
$C(n)\le\frac{(n-1)(n+4)}6$.
\end{lemma}
\begin{proof}
We prove the Lemma using the induction on $n$.
By a computer search, we obtain $C(1)=0,\,C(2)=1,\,C(3)=2,\,C(4)=3,\,C(5)=5,\,C(6)=7$.
Hence, the inequality is satisfied for $n\le 6$.
From Lemma~\ref{lem:mind},
for $n\ge 7$
\begin{align*}
C(n)&\le C(n-1)+2\left\lfloor\frac{n}6\right\rfloor+1+\mathbb{I}\{n\equiv 5\bmod6\}\\
&\le C(n-6)+2\frac{n+(n-1)+(n-2)+(n-3)+(n-4)+(n-5)-1-2-3-4-5}6+7\\
&= C(n-6)+ 2n -3\\
&\le \frac{(n-7)(n-2)}6 + 2n -3\\
&= \frac{(n-1)(n+4)}6.\qedhere
\end{align*}
\end{proof}
Note that if we assume that the minimum distance of self-dual additive $\mathbb{F}_4$ codes is at most $2\lfloor\frac{n}6\rfloor+1$ for $n\equiv 1\bmod 6$, we obtain an upper bound $C(n)\le \frac{n(n+2)-2}6$.
Although the optimal asymptotic CZ-complexity is $C(n)=\Theta(n^2/\log n)$~\cite{markov2008optimal},
Lemma~\ref{lem:wcost} gives the upper bound on the CZ-complexity for finite $n$.

Corollary~\ref{cor:23r} and Lemma~\ref{lem:md} imply an upper bound on the minimum distance of $4^{\mathrm{H}+}$-codes.
\begin{corollary}\label{cor:md}
For any graph $G$ with $n$ vertices and the rank-width $r\le n/3$, the $4^{\mathrm{H}+}$-code $\mathcal{C}_G$ has the minimum distance at most $\lceil (3r+1)/2\rceil$.
\end{corollary}
Corollary~\ref{cor:md} for $r=1$ was shown in~\cite{kumar2019graphs}.

The rank-width of an $n$-vertex graph is at most $\lceil n/3\rceil$.
In other words, any graph with rank-width $r$ has at least $3r-2$ vertices.
As another consequence of Lemma~\ref{lem:mind}, we show that there is no graph with $3r-2$ vertices and rank-width $r$ when $r$ is even.

\begin{lemma}\label{lem:3r2e}
For any even $r\ge 2$, a graph $G$ with $3r-2$ vertices has rank-width at most $r-1$.
\end{lemma}
\begin{proof}
Let $G$ be a graph with $3r-2$ vertices.
From Lemma~\ref{lem:mind}, there is a set $S\subseteq V(G)$ of size $r$ with $\cutrank_G(S)\le r-1$.
Let $A,\, B \subseteq V(G)\setminus S$ be a pair of subsets satisfying $A\cup B = V(G)\setminus S$, $A\cap B = \varnothing$ and $|A|=|B|=r-1$.
Then, we consider three rooted binary trees $T_S$, $T_A$ and $T_B$ where their leaf nodes correspond to $S$, $A$ and $B$, respectively.
Then, we obtain a rank-decomposition of $G$ with a vertex that connects to the roots of $T_S$, $T_A$ and $T_B$.
The width of the rank-decomposition is at most $r-1$.
\end{proof}
In order to generalize Lemma~\ref{lem:3r2e} for odd $r$ by the same proof, we need the conjecture $d\le2\lfloor\frac{n}6\rfloor+1$ for $n\equiv1\bmod6$,
which was confirmed for $n\le 25$.
\begin{lemma}
For $r=1,3,5,7,9$, a graph $G$ with $3r-2$ vertices has rank-width at most $r-1$.
\end{lemma}

\subsection{The proof of Theorem~\ref{thm:upper}}
By combining Lemmas~\ref{lem:54r} and \ref{lem:wcost}, we prove Theorem~\ref{thm:upper}.

\begin{theorem}[(Equivalent to Theorem~\ref{thm:upper})]\label{thm:upperEC}
For any odd $r$ and $n\ge \frac{13r^2-6r-3}{4r}$,
\begin{align*}
C(n,r)&\le \frac{5r^2-1}{4r}n -\frac{221r^4-180r^3+10r^2+36r+9}{96r^2}.
\end{align*}
For any even $r$ and $n\ge \frac{13r-6}4$,
\begin{align*}
C(n,r)&\le \frac{5r}{4}n - \frac{221r^2 - 180r + 100}{96}.
\end{align*}
\end{theorem}
\begin{proof}
Assume $r$ is odd.
We prove the theorem by induction on $n$.
We first prove that if $n\in\left[\frac{13r^2-6r-3}{4r},\,\frac{17r^2-6r-3}{4r}\right]$,
\begin{align*}
C(n,r)&\le \frac{5r^2-1}{4r}n -\frac{221r^4-180r^3+10r^2+36r+9}{96r^2}.
\end{align*}
This can be proved by
\begin{align*}
C(n,r) &\le \frac{(n-1)(n+4)}6\\
&= \frac{5r^2-1}{4r}n +  \frac{(n-1)(n+4)}6 - \frac{5r^2-1}{4r}n\\
&= \frac{5r^2-1}{4r}n + \frac16\left(n^2 - \frac{15r^2-6r-3}{2r}n - 4\right)\\
&\le \frac{5r^2-1}{4r}n - \frac{221r^4-180r^3+10r^2+36r+9}{96r^2}\qquad
\end{align*}
where in the last inequality, we substitute $n=\frac{13r^2-6r-3}{4r}$ into the second term since it maximizes this term.

We now assume $n\ge\lfloor\frac{17r^2-6r-3}{4r}+1\rfloor\ge 3r$.
From Lemma~\ref{lem:54r},
\begin{align*}
C(n, r) &\le \frac{5r^2-1}{4r}k + C(n-k, r)\\
&\le \frac{5r^2-1}{4r}k + \frac{5r^2-1}{4r}(n-k) - \frac{221r^4-180r^3+10r^2+36r+9}{96r^2}\\
&= \frac{5r^2-1}{4r}n - \frac{221r^4-180r^3+10r^2+36r+9}{96r^2}.
\end{align*}

Assume $r$ is even.
We prove the theorem by induction on $n$.
We first prove that if $n\in\left[\frac{13r-6}{4},\,\frac{17r-6}{4}\right]$,
\begin{align*}
C(n,r)&\le \frac{5r}{4}n -\frac{221r^2-180r+100}{96}.
\end{align*}
This can be proved by
\begin{align*}
C(n,r) &\le \frac{(n-1)(n+4)}6\\
&= \frac{5r}{4}n +  \frac{(n-1)(n+4)}6 - \frac{5r}{4}n\\
&= \frac{5r}{4}n + \frac16\left(n^2 - \frac{15r-6}{2}n - 4\right)\\
&\le \frac{5r}{4}n - \frac{221r^2-180r+100}{96}\qquad
\end{align*}
where in the last inequality, we substitute $n=\frac{13r-6}{4}$ into the second term since it maximizes this term.

We now assume $n\ge\lfloor\frac{17r-6}4+1\rfloor \ge 3r$.
From Lemma~\ref{lem:54r},
\begin{align*}
C(n, r) &\le \frac{5r}{4}k + C(n-k, r)\\
&\le \frac{5r}{4}k + \frac{5r}{4}(n-k) - \frac{221r^2-180r+100}{96}\\
&= \frac{5r}{4}n - \frac{221r^2-180r+100}{96}.\qedhere
\end{align*}
\end{proof}

\end{document}